\documentclass[sigconf]{aamas} 

\usepackage{amsmath}
\usepackage{amsthm}
\usepackage{dsfont}
\usepackage{subcaption}
\usepackage{bm}
\usepackage{natbib}

\DeclareMathOperator{\alg}{\textsc{alg}}
\DeclareMathOperator{\E}{\mathbb{E}}
\DeclareMathOperator{\e}{\text{e}}
\renewcommand{\P}{\mathbb{P}}
\DeclareMathOperator{\N}{\mathbb{N}}
\DeclareMathOperator{\gr}{\textsc{gr}}
\DeclareMathOperator{\opt}{\textsc{opt}}
\DeclareMathOperator*{\argmax}{argmax}

\makeatletter
\newtheorem*{rep@theorem}{\rep@title}
\newcommand{\newreptheorem}[2]{%
\newenvironment{rep#1}[1]{%
 \def\rep@title{#2 \ref{##1}}%
 \begin{rep@theorem}}%
 {\end{rep@theorem}}}
\makeatother

\newtheorem{theorem}{Theorem}
\newreptheorem{theorem}{Theorem}
\newtheorem{definition}[theorem]{Definition}
\newtheorem{proposition}[theorem]{Proposition}
\newreptheorem{proposition}{Proposition}
\newtheorem{lemma}[theorem]{Lemma}
\newreptheorem{lemma}{Lemma}
\newtheorem{example}[theorem]{Example}
\newtheorem{remark}[theorem]{Remark}

\usepackage[ruled]{algorithm2e} 

\SetAlFnt{\small}
\SetAlCapFnt{\small}
\SetAlCapNameFnt{\small}
\SetAlCapHSkip{0pt}
\IncMargin{-\parindent}

\usepackage{balance} 

\doi{UDNL2582}

\makeatletter
\gdef\@copyrightpermission{
  \begin{minipage}{0.2\columnwidth}
   \href{https://creativecommons.org/licenses/by/4.0/}{\includegraphics[width=0.90\textwidth]{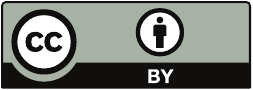}}
  \end{minipage}\hfill
  \begin{minipage}{0.8\columnwidth}
   \href{https://creativecommons.org/licenses/by/4.0/}{This work is licensed under a Creative Commons Attribution International 4.0 License.}
  \end{minipage}
  \vspace{5pt}
}
\makeatother

\setcopyright{ifaamas}
\acmConference[AAMAS '26]{Proc.\@ of the 25th International Conference
on Autonomous Agents and Multiagent Systems (AAMAS 2026)}{May 25 -- 29, 2026}
{Paphos, Cyprus}{C.~Amato, L.~Dennis, V.~Mascardi, J.~Thangarajah (eds.)}
\copyrightyear{2026}
\acmYear{2026}
\acmDOI{}
\acmPrice{}
\acmISBN{}

\title{Diverse Committees \\ with Incomplete or Inaccurate Approval Ballots}
\subtitle{AAAI Track}

\author{Feline Lindeboom}
\affiliation{
  \institution{Rijksuniversiteit Groningen}
  \city{Groningen}
  \country{The Netherlands}}
\email{felinelindeboom@proton.me}

\author{Martijn Brehm}
\affiliation{
  \institution{Universiteit van Amsterdam}
  \city{Amsterdam}
  \country{The Netherlands}}
\email{m.a.brehm@uva.nl}

\author{Davide Grossi}
\affiliation{
  \institution{Rijksuniversiteit Groningen}
    \institution{Universiteit van Amsterdam}
  \city{Groningen}
  \country{The Netherlands}}
\email{d.grossi@rug.nl}

\author{Pradeep K. Murukannaiah}
\affiliation{
  \institution{Technische Universiteit Delft}
  \city{Delft}
  \country{The Netherlands}}
\email{p.k.murukannaiah@tudelft.nl}

\begin{abstract}
We study diversity in approval-based committee elections with incomplete or inaccurate information. We define diversity according to the Maximum Coverage problem, which is known to be $\mathsf{NP}$-complete, with a best attainable polynomial time approximation ratio of $1-1/\e$. In the incomplete information setting, voters vote only on a small portion of the candidates, and we prove that getting arbitrarily close to the optimal approximation ratio w.h.p. requires $\Omega(m^2)$ non-adaptive queries, where $m$ is the number of candidates. This motivates studying adaptive querying algorithms, that can adapt their querying strategy to information obtained from previous query outcomes. In that setting, we lower this bound to only $\Omega(m)$ queries. We propose a greedy algorithm to match this lower bound up to log-factors. We prove the same $\tilde\Theta(m)$ bound for the generalized problem of Maximum Coverage over a matroid constraint, using a local search algorithm. Specifying a matroid of valid committees lets us implement extra structural requirements on the committee, like quota. In the inaccurate information setting, voters' responses are corrupted with a small probability. We prove $\tilde\Theta(nm)$ queries are required to attain a $(1-1/\e)$-approximation with high probability, where $n$ is the number of voters. While the proven bounds show that all our algorithms are viable asymptotically, they also show that some of them would still require large numbers of queries in instances of practical relevance. Using real data from Polis as well as synthetic data, we observe that our algorithms perform well also on smaller instances, both with incomplete and inaccurate information.
\end{abstract}


\newcommand{\BibTeX}{\rm B\kern-.05em{\sc i\kern-.025em b}\kern-.08em\TeX}

\begin{document}

\pagestyle{fancy}
\fancyhead{}

\maketitle 

\section{Introduction}
Consensus has grown among political scientists, democracy practitioners, and decision makers alike, that effective involvement of citizens in policy decisions should be regarded as a high priority for democratic institutions at all levels, from local, to national, to regional \cite{landemore_democratic_2017,matsusaka2020let}.
Currently, it is on digital democracy platforms especially that research efforts are concentrating \cite{brill2018interactive,grossi2024enabling}, in order to provide citizens with effective civic participation tools \cite{mikhaylovskaya2024enhancing}. 

A wealth of these types of digital democracy tools have been developed and deployed around the world in the last decade: LiquidFeedback \cite{behrens2014principles}, Consul, YourPriorities, Decidim, Polis \cite{small2021polis}, to mention a few.\footnote{See \url{www.liquidfeedback.org}; \url{www.consulproject.org}; \url{www.citizens.is}; \url{www.decidim.org}; \url{www.pol.is}.}
In different forms, all these applications allow users to provide free-text input to public deliberations and enable them to express their own opinions on the input of others. Supporting such processes involves the use of algorithms for information-processing problems. One such problem concerns how to effectively summarize the current state of the deliberation: {\em how does `the group' currently think about the issues being deliberated upon?}

One approach to this problem, which is followed for instance in Polis, consists of selecting sets of statements proposed by users, based on how much support the statements have elicited from other users. As Halpern, Kehne, Procaccia, Tucker-Foltz and Wüthrich \cite{halpern_representation_2023} already noticed, this approach can be conceptualized as an approval-based committee election problem \cite{faliszewski2017multiwinner,lackner_multi-winner_2023}, as studied in computational social choice \cite{brandt2016handbook}: users `vote' for other users' statements by expressing whether they approve/support them. Approaching the problem from this point of view has two main advantages. First, it can guarantee summaries that are representative of the discussion in a rigorous sense. Second, properly designed algorithms can enhance transparency of the decision-making process, which can in turn increase participants' trust in this same process \cite{mikhaylovskaya_building_2024}. Vice-versa a non-transparent process may generate distrust, while individuals with more political distrust tend to show less interest in politics and lower rates of civic participation \cite{eder_political_2015}. 

Unlike in standard approval based committee elections, information about who supports which statements is now sparse, as no user can possibly express whether they agree or disagree with every single statement contributed by their peers. For the same reason, voters can become inaccurate in representing their true beliefs, when having to answer many questions. Such a deliberation summarization problem can thus be thought of as an approval-based committee election problem in which ballots are incomplete or inaccurate, or both. From this perspective, \citet{halpern_representation_2023} focused on the problem of designing selection algorithms yielding summarizations that meet specific forms of proportional representation proposed in the computational social choice literature \cite{aziz2017justified}, while querying users' opinions as efficiently as possible.

Our work builds on the approach put forth by \citet{halpern_representation_2023}, but focuses on the property of diversity instead. Diversity has received less attention than proportional representation in the committee election literature but, we argue, is an important objective for deliberation summarization, as already mentioned in \citet{lackner_multi-winner_2023}. {\em First}, diverse summaries can guarantee higher inclusivity, an attribute one could desire in and of itself. {\em Second}, an inclusive summary can show more participants that their input is taken into account. This increases people's trust in the system at hand \cite{mikhaylovskaya2024enhancing}, which motivates to participate also in future instances. {\em Third}, a diverse summary gives a broad, informative view of `the group's' opinion. This is especially valuable in deliberation summarization contexts, where the number of candidates may be very large. {\em Fourth}, cognitive diversity is argued, e.g., by \citet{landemore_democratic_2017}, to also be of epistemic value from a crowd-wisdom perspective. This is relevant because the selected statements typically serve as a basis for further discussions. 

We use the established formalization of diversity in committee elections based on the Chamberlin-Courant score \cite{chamberlin_representative_1983} and the Maximum Coverage problem \cite{nemhauser_analysis_1978,skowron_fully_2015}. Our focus thus lies on the computational problem of constructing summarizations that maximize the coverage of voters, while querying users' opinions as efficiently as possible.

\section{Related work and contribution}\label{sec:related_work}

\paragraph{Diverse committees
and the Maximum Coverage problem.
}
In committee election theory, the diversity of a committee is commonly quantified by the \textit{Chamberlin-Courant} score \cite{chamberlin_representative_1983}. In \emph{approval-based} committee elections, a voter's \emph{representative} in an elected committee $W$ then is any one candidate in $W$ they approve of, and the \emph{approval-based Chamberlin-Courant} (going forward: Chamberlin-Courant) score of $W$ equals the number of voters that have a representative in $W$. 
Portraying candidates by their set of approving voters, we see this problem amounts to selecting a set of candidates that maximizes the number of `covered' voters: the approval-based Chamberlin-Courant problem is equivalent to the unit-weight Maximum Coverage problem.

An instance of the unit-weight Maximum Coverage problem, Max Cover going forward, consists of a set $V$ of $n$ elements $v_i$, 
a family $C\subseteq\mathcal{P}(V)$ of $m$ subsets $c_j$, and a natural number $k$. The goal is to select $k$ subsets so that 
$|\bigcup_{j=1}^k c_j|$ is maximized. 
\citet{hochbaum_analysis_1998} proved, by reduction from Set Cover, that the Max Cover decision problem is $\mathsf{NP}$-complete, and that a greedy algorithm is $(1-1/\e)$-approximate. Shortly after, \citet{feige_threshold_1998}, using a reduction from approximating Max 3SAT-5,  proved that for any $\epsilon>0$, Max Cover cannot be approximated in polynomial time within a ratio of $1-1/\e+\,\epsilon$, unless $\mathsf{P} = \mathsf{NP}$; the simple greedy algorithm is thus optimal. Many years later, Cohen-Addad, Gupta, Kumar, Lee and Li \cite{cohenaddad2019tightfptapproximationskmedian} showed that under the Gap Exponential Time Hypothesis\footnote{For $k\ge 3$, define $s_k:=\inf\{\delta:\exists \text{ an algorithm that solves } k\text{-SAT in } 2^{\delta n}\text { time}\}$. ETH states that for $k\ge 3,\, s_k>0$; $k$-SAT cannot be solved in subexponential time, for $k\ge 3$. ETH was first formulated by \citet{impagliazzo_complexity_2001} and would imply that $\mathsf{P}\neq \mathsf{NP}$.}, 
for any $\epsilon>0$, there exists no FPT-approximation algorithm for parameter $k$ that approximates Max Cover within factor $1-1/\e+\,\epsilon$. \citet{manurangsi_tight_2019} improved the running time lower bound from $f(k)\cdot (m+n)^{k^{\text{poly}(1/\epsilon)}}$ for any function $f$ and $\epsilon>0$, to $f(k)\cdot (m+n)^{o(k)}$, which is tight. Essentially, this shows a brute-force approach going over all size-$k$ subsets is the best we can do if we insist on obtaining an approximation ratio better than $1-1/\e$.

\citet{peters_single-peakedness_2018} proved that maximizing the Chamberlin-Courant score can be efficiently done on the Candidate Interval (CI) domain, an approval-based version of the Single-Peaked domain. \citet{peters_circle_2020} extended this result to circular preference domains and Sornat, Williams and Xu \cite{sornat_near-tight_2022} extended it to nearly CI domains. \citet{skowron_fully_2015} made boundedness assumptions on the number of sets an element appears in. When any element appears in \emph{at most} $p$ sets, optimizing over the largest $\lceil\frac{2pk}{1-\beta}+k\rceil$ sets yields a $\beta$-approximate solution in FPT time (in $k$). When any element appears in \emph{at least} $p$ sets, the greedy algorithm achieves an improved approximation ratio of $1-\e^{-\max\{\frac{pk}{m},1\}}$. Finally, a generalization of the problem was studied by \citet{filmus_power_2013}, who presented a local search algorithm that is $(1-1/\e)$-approximate for Max Cover over a matroid constraint. The greedy algorithm attains only ratio $1/2$ in this generalized setting.

\paragraph{Incomplete or inaccurate information.}
Single-winner voting with incomplete information received some attention over the years in the computational social choice literature \cite{boutilier2016incomplete}. More  recent work has addressed the issue of achieving proportional representation in committee elections under incomplete information \cite{imber2022approval,halpern_representation_2023,fish2024generative}, or of maximizing arbitrary scoring rules for ranking candidates, subject to the satisfaction of the Justified Representation criterion \cite{revelrepresentative}. Among these contributions, the work of
\citet{halpern_representation_2023} is closest to what is our focus  in this paper. The authors study proportionality in approval-based committee elections under incomplete information. The paper presents a version of local search Proportional Approval Voting that queries voters and, using $O(mk^6\log k\log m)$ queries, finds a solution that satisfies Extended Justified Representation \cite{aziz2017justified} and Optimal Average Satisfaction (based on the notion of average satisfaction \cite{Sanchez-Fernandez_Elkind_Lackner_Fernandez_Fisteus_Basanta_Val_Skowron_2017}) with high probability. The result extends to an $\alpha$-approximate version of both axioms, in which case the query complexity decreases by a factor $k^3$. The same study presents a lower bound of $\Omega(m^{11})$ for the query complexity of algorithms that cannot adapt their querying strategy to information obtained from previous queries (non-adaptive algorithms). 

\paragraph{Our contribution.}
We study diversity in approval-based committee elections when information elicited from voters is incomplete or inaccurate. We measure the diversity of a solution with its Max Cover, or equivalently, Chamberlin-Courant, score. For a large part, our work combines the two lines of research outlined above, and is based in particular on the work of \citet{filmus_power_2013} and \citet{halpern_representation_2023}. We make three main contributions.

{\em First}, in a setting with incomplete information, we prove that getting arbitrarily close to the optimal approximation ratio with high probability (w.h.p.) requires $\Omega(m^2)$ non-adaptive queries (this result is presented only in the full version of the paper\footnote{\url{https://arxiv.org/abs/2506.10843}}). This motivates studying \textit{adaptive} querying algorithms, that can adapt their querying strategy to information obtained from previous query outcomes. In that setting, we lower this bound to only $\Omega(m)$ queries. We adapt the greedy algorithm to match this lower bound up to log-factors (Theorem~\ref{thm:greedy_bound}). We prove the same $\tilde\Theta(m)$\footnote{ The notation $\tilde\Theta(\cdot)$ and similarly for $\Omega, O$, suppresses terms of the form $\log^{O(1)}(n)$ where $n$ is the growing parameter.} bound for the generalized problem of Max Cover over a matroid constraint
(Theorem~\ref{thm:ls_bound}). Specifying a matroid of valid committees lets us implement external diversity requirements, like upper and lower quota on groups of candidates. {\em Second}, in the inaccurate information setting, we prove that recovering the optimal approximation ratio w.h.p. requires $\tilde\Theta(nm)$ queries (Theorem~\ref{thm:inaccurate_lower_bound}). Despite these positive asymptotic bounds, summarized in Table~\ref{tab:tree}, our results do involve sizeable constant overheads for some of our algorithms, which make viability on real-world instances questionable. So, {\em third}, using real data from Polis and synthetic data that we produced by adapting methods from Szufa, Faliszewski, Janeczko, Lackner, Slinko, Sornat and Talmon \cite{szufa_how_2022}, we empirically show that our algorithms perform considerably better than our worst-case analysis suggests. Importantly, this appears to hold even in situations in which votes are both incomplete {\em and} inaccurate. All proofs are provided in the full version of the paper.

\renewcommand{\arraystretch}{1.37}
\begin{table}[]
  \caption{Overview of query complexity bounds in $m$ and $n$.}
\begin{tabular}{cccccccc}
 \multicolumn{4}{c|}{Incomplete} & \multicolumn{4}{c}{Inaccurate} \\  \cline{5-8}
 \multicolumn{2}{c}{Adaptive} & \multicolumn{2}{|c|}{Non-adaptive} & \multicolumn{4}{c}{$\tilde\Theta(nm)$} \\ \cline{3-4} 
 Matroid & \multicolumn{1}{|c|}{No matroid} & \multicolumn{2}{c}{$\Omega(m^2)$} \\ \cline{1-2} 
 $\tilde\Theta(m)$ & $\tilde\Theta(m)$
 \end{tabular}
 \label{tab:tree}
 \end{table}

\section{Preliminaries}\label{sec:prelim}
An instance of an approval-based committee election problem consists of a set $V$ of $n$ voters, a set $C$ of $m$ candidates, and a natural number $k$. The goal is to elect a committee $W\subseteq C$ of size $k\le m$, based on the opinions of the voters in $V$. An instance also contains approval information: every voter approves of a subset of the candidates. Usually, this approval information is known upfront and expressed in the form of an \emph{approval set} $A(i)$ for each voter $v_i\in V$. The sequence of sets $A=(A(1),\ldots,A(n))$ is called the \emph{approval profile}. We call this setting the \emph{perfect information setting}. This paper studies incomplete and inaccurate information settings; the respective models are detailed in the corresponding section. To measure the diversity of a committee, we use the Chamberlin-Courant score.
\begin{definition}[\citet{chamberlin_representative_1983}]
On any approval-based committee election instance $(V,C,A,k)$, the \emph{Chamberlin-Courant (CC)} score of a committee $W\subseteq C,\,|W|=k$ is
\begin{align}
\text{CC}(W)=\frac{1}{n}\sum_{i=1}^n\vmathbb{1}_{\{A(i)\,\cap\, W\,\neq\,\emptyset\}}(i).
\end{align}
In the Chamberlin-Courant decision problem, the input is an instance $(V,C,A,k)$ and an integer $x$, and the question is whether it is possible to achieve score $\frac{x}{n}$ using $k$ candidates. In the Chamberlin-Courant problem, the input is an instance $(V,C,A,k)$ and the task is to maximize the score using $k$ candidates.
\end{definition}


As previously mentioned, the approval-based Chamberlin-Courant problem is equivalent to the unit-weight Maximum Coverage problem. For the rest of the paper, we use the terminology of the Chamberlin-Courant optimization problem, with voters, candidates and committees (as opposed to elements, sets and covers).

\section{Incomplete information}
In the \textit{incomplete information setting}, the approval profile $A=(A(1),\dots,A(n))$ is replaced by a function $A(i,j):=\vmathbb{1}_{\{c_j\in A(i)\}}(i,j)$ which equals 1 in case voter $v_i$ approves candidate $c_j$ and 0 otherwise. This is a generalization of the original model, as querying all voters about all candidates recovers the complete approval profile. As such, by making $nm$ queries to $A$ we can
find a $(1-1/\e)$-approximate solution, e.g. with a greedy approach. In the full version of the paper, we show that any non-adaptive querying algorithm must make $\Omega(m^2)$ queries to $A$ to get arbitrarily close to this optimal approximation ratio w.h.p., where non-adaptive means that the sequence of queries made must be specified beforehand. In Polis data we observe that $m \in \Theta(n)$, which would imply $\Omega(m^2)=\Omega(nm)$, making no improvement over the complete information setting. We can do better by studying adaptive querying algorithms, that can adapt their querying strategy to information obtained from previous query outcomes. For such adaptive algorithms we show that obtaining an approximation ratio arbitrarily close to this optimal ratio w.h.p. requires only $\tilde\Theta(m)$ queries to $A$.

To see that we need at least $\Omega(m)$ queries, imagine an instance where only one candidate has any approvals. Any algorithm must select this candidate to attain a positive approximation ratio. In the worst-case, finding this candidate w.h.p. requires sampling voters for each of the $m$ candidates. In the remainder of the section we study algorithms that witness this lower bound $\Omega(m)$ up to log-factors. We first adapt a greedy algorithm for CC and then do the same for a local search algorithm for the generalized problem of optimizing over a matroid constraint. Both algorithms achieve an approximation ratio arbitrarily close to the optimal ratio w.h.p. using $\tilde O(m)$ queries, matching the lower bound (up to log-factors).

\subsection{Upper bound for the unconstrained setting}
\paragraph{The standard greedy algorithm.}
We prepare the ground by considering the greedy algorithm for the perfect information setting. See Algorithm~\ref{alg:greedy}, where we write $\Delta(W,c)$ for the increase in CC-score obtained by adding candidate $c$ to committee $W$. Algorithm~\ref{alg:greedy} elects, in each iteration, the candidate that yields the largest immediate increase in CC-score, and achieves the optimal approximation ratio of $1-1/\e$ \cite{hochbaum_analysis_1998}. We give a novel proof of its approximation ratio in the full version of this paper.


\begin{algorithm}[b]
    \caption{\textsc{greedy}}
    \label{alg:greedy}
	\SetAlgoNoLine
	\KwIn{Numbers $n,m,k\in\N$ with $k\le m$, set $V$ of $n$ voters, set $C$ of $m$ candidates.}
	\KwOut{Committee $W\subseteq C$ of size $k$.}
    Let $W=\{\}$ be an empty set\;
    \For{$i=1,\ldots,k$,}{
        Add $c'\in \argmax_{c\notin W}\,\Delta(W,c)$ to $W$.}
\end{algorithm}

\paragraph{The greedy query algorithm.}
We adapt Algorithm~\ref{alg:greedy} to the incomplete information setting. Our strategy is to repeatedly sample sets of voters (of size $\ell\le n$) uniformly at random and query them about a subset of candidates of size $t\le m$, in order to estimate $\Delta(W,c)$ for all $c\notin W$. The algorithm then behaves like the standard greedy algorithm: iteratively selecting the candidate $c$ maximizing this estimate of $\Delta(W,c)$. By carefully bounding the deviation from the true value $\Delta(W,c)$, we can guarantee that our algorithm attains a ratio arbitrarily close to the optimal approximation ratio w.h.p.

To write this more formally, note that if we ask \emph{all} voters for their votes on some query set $Q\subseteq C$ (of size $t$), the responses allow us to compute, for any set $S\subseteq Q$:
\begin{align}
p_S := \frac{1}{n}\sum_{i=1}^n\vmathbb{1}_{\{A(i)\,\cap\, S\,\neq\, \emptyset\}}(i).
\end{align}
The value $p_S$ equals the Chamberlin-Courant score of $S$, and in this case $\Delta(W,c) = p_{W\cup \{c\}}-p_{W}$, so we can compute $\Delta(W,c)$ with complete information on a query set containing $W\cup\{c\}$. Sampling $\ell$ voters uniformly at random and querying them about the set $Q\subseteq C$, we can compute, for any subset $S\subseteq Q$,
\begin{align}
\hat{p}_S := \frac{1}{\ell}\sum_{i=1}^\ell\vmathbb{1}_{\{A(i)\,\cap\, S\,\neq\,\emptyset\}}(i),
\end{align}
that estimates $p_S$. We write $\hat{\Delta}(W,c):=\hat{p}_{W\cup \{c\}}-\hat{p}_{W}$ for an estimate of $\Delta(W,c)$ based on $\hat{p}$. Using this approach, we obtain Algorithm~\ref{alg:greedy-incomplete-queries}.

In step 1 of the for-loop, to be able to calculate the values $\Delta(W,c')$, we add the committee to each query set and, other than that, distribute the candidates over the query sets until they reach size $t$, without further restrictions.

For Theorem~\ref{thm:greedy_bound}, we assume that $\sum_{v_i\in V}\vmathbb{1}_{\{|A(i)|\ge 1\}}(i)\ge k$, that is: there exist at least $k$ voters with a non-empty approval set.\footnote{We use this assumption to prove the approximation guarantee of the greedy query algorithm, see also the full version of the paper. The assumption is also motivated practically: instances where $\sum_{v_i\in V}\vmathbb{1}_{\{|A(i)|\ge 1\}}(i)< k$ are uninteresting in the context of deliberation summarization---where the challenge is to aggregate \emph{many} different opinions into a small summary---and do not occur in practice.} 
Theorem~\ref{thm:greedy_bound} gives an upper bound on the number of queries required to run Algorithm~\ref{alg:greedy-incomplete-queries} on such instances.

\begin{algorithm}[t]
    \caption{\textsc{greedy-incomplete}}
    \label{alg:greedy-incomplete-queries}
	\SetAlgoNoLine
	\KwIn{Numbers $n,m,k\in\N$ with $k\le m$, set $V$ of $n$ voters, set $C$ of $m$ candidates, query size $t$ with $m\ge t>k$, and $\gamma\in(0,1),\ \delta>0$.}
	\KwOut{Committee $W\subseteq C$ of size $k$.}
    Set $\epsilon=\frac{(1-\gamma)e}{\gamma(e-1)}$ and $\ell=\lceil \frac{2}{\epsilon^2}(\log(\frac{2mk}{\delta}))\rceil$\;
    Let $W=\{\}$ be an empty set\;
    \For{$i=1,\ldots,k$,}{
		\begin{enumerate}
			\item Construct the smallest set of query sets $\mathcal{Q}=\{Q_i\}_i$ with $Q_i\subseteq C$ and $|Q_i|=t$ for all $i$, and such that $W\subseteq\bigcap\mathcal{Q}$ and $\bigcup\mathcal{Q}=C$, and present each query set to $\ell$ voters sampled uniformly at random.
			\item  For all $c\notin W$, determine $\hat{\Delta}(W,c)$ as an estimate of $\Delta(W,c)$ using $\ell$ voters and query set $Q$ containing $W\cup c$.
			\item  Add  $c'\in\argmax_{c\notin W}\,\hat{\Delta}(W,c)$ to $W$.
		\end{enumerate}
	}
\end{algorithm}

\begin{theorem}\label{thm:greedy_bound}
Let $\sum_{v_i\in V}\vmathbb{1}_{\{|A(i)|\ge 1\}}(i)\ge k$, $\delta>0,\, \gamma\in(0,1)$, and $k<t\le m$. Then, w.p. at least $1-\delta$, Algorithm~\ref{alg:greedy-incomplete-queries} is $(1-1/e)\gamma$-approximate for CC with query complexity  
$$O\left(\left(\frac{\gamma}{1-\gamma}\right)^2km\log\left(\frac{km}{\delta}\right)\right)\in O_{\delta,\gamma,k}(m\log m).$$
\end{theorem}

A few observations are in order. First note that, asymptotically, the query complexity provided by Theorem \ref{thm:greedy_bound} is much smaller than querying the entire profile of size $nm$. The constant overhead, which increases with $k$ and $\gamma$ and decreases with $\delta$, is quite small. 
Second, note that a voter may be queried more than once during the run of the algorithm, because we sample with replacement. However, since the query complexity is independent of $n$, sampling with replacement approaches sampling without replacement as $n$ grows.
Third, the asymptotic complexity does not depend on $t$. Increasing $t$ means fewer voters get larger query sets, which decreases the total number of queries only slightly (by a small constant factor). In practice, the value of $t$ should be chosen large enough so that the query sets leave room for sufficiently many unelected candidates, but not too large, since we assume that voters are not able to express their opinion on all $m$ candidates.

To give further context to Theorem \ref{thm:greedy_bound}: in real-world instances of online deliberations such as those ran through Polis, we observe that typically $m\in\Theta(n)$ and $k\ll m$. So, suppose that $m=1000$, $\, \gamma=0.85,\,\delta=0.05,\,k=8$ and $t=20$, then Theorem~\ref{thm:greedy_bound} requires 432 920 queries. Starting with $n>432$ voters this is smaller than $nm$.
\subsection{Upper bound for a setting with matroid constraints}
\paragraph{Matroid constraints for committee elections.}
We can imagine situations where we want to place restrictions on the set of possible committees to be elected. Perhaps we want to include extra diversity requirements, such as `the committee must contain at least/at most $x$ candidates of category $Y$'. In practice, $x$ would be a natural number and $Y$ could be a demographic group. In online deliberation settings, these attributes would extend to users. 
This provides the possibility to add an extra `dimension' of diversity: where before, the concept of diversity was based solely on the approval profile, we now consider \emph{external attributes} of the candidates. For such purposes, we make a selection of which committees are `allowed' to be elected, and collect them in a set $\mathcal{I}$. It is not obvious that any algorithm adapts well to this restriction. Fortunately, when $\mathcal{I}$ defines a matroid, we can mend this problem \cite{filmus_power_2013}, and, as shown by Masařík, Pierczyński and Skowron \cite{masarik_generalised_2024}, matroids can indeed also be used to implement both upper and lower quota on any number of disjoint categories. See the full version of the paper for the construction.

\begin{definition}[matroid]\label{matroid}
A \textit{matroid} $\mathcal{M}$ is a pair $(C,\mathcal{I})$ where $C$ is a finite universe and $\mathcal{I}$ is a collection of subsets of $C$ (called independent sets) satisfying the following properties:
\begin{enumerate}
	\item $\emptyset\in\mathcal{I}$,
	\item if $A\in\mathcal{I}$ and $B\subset A$ then $B\in\mathcal{I}$, and
	\item for all $A,B\in\mathcal{I}$ with $|A|>|B|$, there exists $x\in A\setminus B$ such that $B\cup\{x\}\in\mathcal{I}$.	
\end{enumerate}
\end{definition}
Property 3) guarantees that all maximal independent sets have the same cardinality. We call such sets \emph{bases}, and their common cardinality is called the \emph{rank} of the matroid.

\begin{example}[uniform matroid]
Consider $\mathcal{I}=\{W\subseteq C:|W|\le k\}$. Then $\mathcal{M}=(C,\mathcal{I})$ is the uniform matroid of rank $k$: 1) $|\emptyset|=0\le k$, 2) for any set $W\in\mathcal{I}$ with $|W|=j\le k$, for any subset $W'\subset W$, $|W'|<|W|\le k$, and 3) for any two feasible sets $W,\ W'\in\mathcal{I}$ with $|W'|<|W|$, we will have, for any $c\in W\setminus W'$, that $|W'\cup \{c\}|\le k$.
\end{example}
Thus, taking for $\mathcal{M}$ the uniform matroid of rank $k$, we retrieve the original problem: the problem with matroid constraints is a generalization of the original problem.

\paragraph{The standard non-oblivious local search algorithm.}
We start with the perfect information setting. A classical (or oblivious) local search algorithm starts with an arbitrary solution and, in each iteration, swaps one or multiple of the elected  elements for unelected elements in order to improve the objective function value. It thus \textit{searches} in the \textit{local} neighborhood of the solution. It stops when no local improvement is possible. At worst, this occurs after all exponentially many options have been exhausted. With the use of approximate local search, this running time problem can be resolved at a small cost in the approximation ratio \cite{orlin_approximate_2004}. In practice, this means we terminate the algorithm when we encounter an improvement smaller than some parameter $\beta>0$\footnote{Using a partial enumeration technique, we could even eliminate this small cost again, see Calinescu, Chekuri, Pál and Vondrák \cite{calinescu_maximizing_2011}.}.

Over a matroid constraint, the greedy algorithm attains only approximation ratio $\frac{1}{2}<1-1/\e$. The approximation ratio of oblivious local search is $\frac{k-1}{2k-\ell-1}$ when $\ell$ sets are exchanged in each iteration \cite{filmus_power_2013}. For $k=2$ and $\ell=1$, this also equals $\frac{1}{2}$. \citet{filmus_power_2013} thus switch to a non-oblivious local search algorithm that uses an auxiliary objective function $f$ for the iterative procedure. The function adds temporary weight to elements covered multiple times, the idea being that elements covered multiple times, remain covered after the next exchange. Such a function can thus prevent getting stuck in bad local optima. We denote by $\alpha_j$ the temporary weight associated to an element covered $j$ times, and write $h_i(W)$ for the number of times that $v_i$ is covered by $W$. The auxiliary objective function is then defined as
\begin{align}
f(W)=\frac{1}{n}\sum_{i=1}^n\alpha_{h_i(W)}.
\end{align}
Note that, with $\alpha_0=0$ and $\alpha_j=1$ for all $j>0$, we retain the original oblivious objective function, but if we set $\alpha_j>\alpha_1$ for $j>1$, we add additional weight to elements covered multiple times. We define
\begin{align*}
\alpha_0=0,\ \ \ \ \alpha_1=1-\frac{1}{\e},\ \ \ \ \alpha_{j+1}=(j+1)\alpha_j-j\alpha_{j-1}-\frac{1}{\e}.
\end{align*}
This choice for the sequence $(\alpha_n)_{n\in\N_{0}}$ is optimal and, for any $\gamma\in(0,1)$, with this objective function, the non-oblivious local search algorithm, with parameter $\beta$ (decreasing in $\gamma$) is $(1-1/\e-\gamma)$-approximate and runs in polynomial time \cite{filmus_power_2013}.

We adjust the non-oblivious local search algorithm of \citet{filmus_power_2013} to suit our setting better, see Algorithm~\ref{alg:LS}, where for $c\in W$ and $c'\notin W$, we now write $\Delta(W,c',c):=f((W\cup \{c'\})\setminus \{c\})-f(W)$. We elaborate on the adaptations in the full version of the paper, but most importantly, the changes do not affect the approximation guarantees of the algorithm:
for $\beta=C_1\frac{\gamma}{k\log k}$, Algorithm~\ref{alg:LS} is $(1-1/\e-\gamma)$-approximate for any $\gamma\in(0,1)$ and some universal constant $C_1 
\le \frac{\log i}{\alpha_i(1-1/\e-\gamma)}$ for any $i\le k$ (see also \cite[Corollary 6]{filmus_power_2013}).

\begin{algorithm}[t]
    \caption{\textsc{local search}-$\beta$}
    \label{alg:LS}
	\SetAlgoNoLine
	\KwIn{Numbers $n,m,k\in\N$ with $k\le m$, set $V$ of $n$ voters, set $C$ of $m$ candidates, $\beta>0$, matroid $\mathcal{M}$.}
	\KwOut{Committee $W\subseteq C$, of size $k$.}
    Choose $W\subseteq C$ such that $|W|=k$, and $c\in W$ and $c'\notin W$ so that $(W\cup \{c'\}) \setminus \{c\}\in\mathcal{I}$\;
    \Repeat{$\Delta(W,c',c)\le\beta$}{
    \begin{enumerate}
          \item $W=(W\cup \{c'\}) \setminus \{c\}$.
	       \item Let $\mathcal{E}$ be the set of all valid exchanges for $W$ according to $\mathcal{M}$.
	       \item Pick $(c',c)\in\argmax_{(x,y)\in \mathcal{E}}\Delta(W,x,y)$.
    \end{enumerate}
    }
\end{algorithm}

\paragraph{The non-oblivious local search query algorithm.}
We turn to the incomplete information setting. Recall that $A(i,j)$ equals 1 when voter $v_i$ approves of candidate $c_j$ and 0 otherwise, and that we write $A(i)$ for the set of candidates approved by voter $v_i$. Given a query set $Q\subseteq C$, for any $S\subseteq Q$, we redefine
\begin{align}
p_S :=\frac{1}{n}\sum_{i=1}^n\vmathbb{1}_{\{A(i)\,\cap\, Q\,=\,S\}}(i);
\end{align}
the probability that a uniformly chosen voter approves, of all the candidates in $Q$, exactly the set $S$. This is different from how we defined $p_S$ for Algorithm~\ref{alg:greedy-incomplete-queries}
: Algorithm~\ref{alg:greedy-incomplete-queries} requires, for any voter $v_i$, knowing just whether they have a representative in set $S\subseteq Q$, whereas Algorithm~\ref{alg:LS-incomplete-queries} needs the number of representatives. With $p_S$ for $S\subseteq Q$, we can compute $f(Q)$:
\begin{align*}
f(Q)&= \frac{1}{n}\sum_{i=1}^n\alpha_{h_i(Q)}
= \frac{1}{n}\sum_{i=1}^n\sum_{S\subseteq Q}\vmathbb{1}_{\{A(i)\,\cap\,Q\,=\,S\}}(i)\cdot\alpha_{|S|}  \\
&= \sum_{S\subseteq Q}\frac{1}{n}\sum_{i=1}^n\vmathbb{1}_{\{A(i)\,\cap\, Q\,=\,S\}}(i)\cdot\alpha_{|S|}
= \sum_{S\subseteq Q}p_S\cdot \alpha_{|S|}.
\end{align*}
We can compute $\Delta(W,c',c)$ with complete information on a query set containing $W\cup\{c'\}$. When information is incomplete, we estimate: sampling $\ell\le n$ voters uniformly at random and presenting them with query set $Q\subseteq C$, we can compute, for every subset $S\subseteq Q$,
\begin{align}
\hat{p}_S := \frac{1}{\ell}\sum_{i=1}^\ell\vmathbb{1}_{\{A(i)\,\cap\, Q\,=\,S\}}(i),
\end{align}
that estimates $p_S$. We write $\hat{\Delta}(W,c',c) = \hat{f}((W\cup\{c'\})\setminus\{c\})-\hat{f}(W) =\sum_{S\subseteq (W\cup\{c'
\})\setminus\{c\}}\hat{p}_S\cdot \alpha_{|S|}-\sum_{S\subseteq W}\hat{p}_S\cdot \alpha_{|S|}$ for an estimate of $\Delta(W,c',c)$ based on the values $\hat{p}$. Algorithm~\ref{alg:LS-incomplete-queries} is our local search query algorithm. 

Theorem~\ref{thm:ls_bound} assumes $|c_j|\ge 1$ for all $j$ (all candidates have at least one approval) and $k\ge 3$\footnote{ These conditions are sufficient to guarantee the claimed approximation ratio, and dropping them complicates the proofs unnecessarily, since both are very natural assumptions: candidates with no approvals can be removed since their election cannot influence the score, and committees of size $k=1,2$ are arguably uninteresting in the context of deliberation summarization. See also the full version of the paper.}. It gives an upper bound on the number of queries required to run Algorithm~\ref{alg:LS-incomplete-queries} on such instances. As expected, this is higher than the bound stated in Theorem~\ref{thm:greedy_bound}.

\begin{algorithm}[t]
    \caption{\textsc{ls}-$\beta$-\textsc{incomplete}}
    \label{alg:LS-incomplete-queries}
	\SetAlgoNoLine
	\KwIn{Numbers $n,m,k\in\N$ with $k\le m$, set $V$ of $n$ voters, set $C$ of $m$ candidates, $\beta>0$, query size $t$, with $m\ge t>k$, constants $\delta>0$ and $\xi\ge 1$, matroid $\mathcal{M}$.}
	\KwOut{Committee $W\subseteq C$, of size $k$.}
    $\epsilon = \frac{\xi - 1}{2\xi}\cdot \beta$, $\ell=\left\lceil\frac{(2-2/\e)^2}{2\epsilon^2}\log\left(\frac{2\cdot(m-k)k\cdot\xi\alpha_k}{\delta\cdot\beta}\right)\right\rceil$\;
    Choose $W\subseteq C$ such that $|W|=k$, and $c\in W$ and $c'\notin W$ so that $(W\cup \{c'\}) \setminus \{c\}\in\mathcal{I}$\;
    \Repeat{$\hat{\Delta}(W,c',c) < \beta -\epsilon$}{
    \begin{enumerate}
           \item $W=(W\cup \{c'\}) \setminus \{c\}$;
	       \item Let $\mathcal{E}$ be the set of all valid exchanges for $W$ according to $\mathcal{M}$.
          \item Construct the smallest set of query sets $\mathcal{Q}=\{Q_i\}_i$ with $Q_i\subseteq C$ and $|Q_i|=t$ for all $i$, and such that $W\subseteq\bigcap\mathcal{Q}$ and $\bigcup\mathcal{Q}=C$, and present each query set to $\ell$ voters sampled uniformly at random.
	       \item For all $c\in W,\, c'\notin W$, determine $\hat{\Delta}(W,c',c)$ as an estimate of $\Delta(W,c',c)$ using $\ell$ voters and query set $Q$ containing $W\cup \{c'\}$.
	       \item Pick $(c',c)\in\argmax_{(x,y)\in \mathcal{E}}\hat{\Delta}(W,x,y)$.
    \end{enumerate}
    }
\end{algorithm}

\begin{theorem}\label{thm:ls_bound}
Let $|c_j|\ge 1\ \forall j$, $\delta>0$, $\gamma\in(0,1)$, $m\ge t > k\ge 3$. Fix $\beta = C_2\frac{1-\gamma}{\gamma k\log k}$ for some constant $C_2$. Then, w.p. at least $1-\delta$, Algorithm~\ref{alg:LS-incomplete-queries} is $(1-1/e)\gamma$-approximate for CC with query complexity
$ O\left(\left(\frac{\gamma k\log k}{1-\gamma}\right)^3m\log\left(\frac{m\gamma k^2\log k}{\delta(1-\gamma)}\right)\right)\in O_{\delta,\gamma,k}(m\log m)$.
\end{theorem}
Like with \textsc{greedy-incomplete}, the query complexity does not depend on $n$, and is asymptotically much smaller than querying the entire ballot. Unlike before however, there is sizeable constant overhead increasing with $k$ and $\gamma$, and decreasing with $\delta$. Especially the dependence in $k$ is much stronger than before. For the same example, $m=1000,\, \gamma=0.85,\,\delta=0.05,\,k=8$ and $t=20$, Theorem~\ref{thm:ls_bound} requires $9.52 180\cdot 10^{11}$ queries. Only for $n>9.52180\cdot 10^8$ this is smaller than $nm$.

\section{Inaccurate information}
In the inaccurate information setting, voters do hand in complete ballots, but the voters are inaccurate in reporting their true approvals: we assume each query $A(i,j)$ is incorrect with a probability $p$. We model this by defining, for $p\in(0,\frac{1}{2})$, a \emph{$p$-inaccurate query} as $A_p(i,j):=A(i,j) \oplus X$, where $X\sim\text{Bernoulli}(p)$\footnote{This is essentially  the classification noise model from PAC learning \cite{kearns1994introduction} applied to our setting.}.
Samples of $X$ are independent between voters, between candidates, and between consecutive draws of the same query. If $p=0$, we would have accurate information, and with $nm$ queries we would find a $(1-1/\e)$-approximate solution (e.g. using the greedy algorithm). When answers may be corrupted, we can simply pose every question multiple times to compensate for the uncertainty. We do this to obtain an upper bound on the number of queries required to still acquire a $(1-1/\e)$-approximate solution with high probability, adapting Algorithm~\ref{alg:greedy} (see the full version of the paper). Moreover, we provide a matching lower bound (up to log-factors), using a result from multi-armed bandit theory, as done in Theorem 9 in \cite{schafer_complexity_2024}. Both results also hold for the problem of optimizing over a matroid constraint.

\begin{theorem}\label{thm:inaccurate_upper_bound}
Let $p\in(0,\frac12)$, $\delta>0$, $n,m\in\mathbb{N}$. Then there exists an algorithm that is $(1-1/e)$-approximate for CC in the $p$-inaccurate query model w.p. at least $1-\delta$ and with query complexity $O(nm\log(nm/\delta))$.
\end{theorem}
\begin{theorem}\label{thm:inaccurate_lower_bound}
Let $p\in(0,\frac12)$, $\delta>0$, $n,m\in\mathbb{N}$. Then any algorithm that is $(1-1/e)$-approximate for CC in the $p$-inaccurate query model w.p. at least $1-\delta$ has expected query complexity $\Omega(nm\log(1/\delta))$.
\end{theorem}
The two bounds match (up to log-factors). Observe that, unlike before, we assume access to the entire ballot, making the bound depend on $n$.

\section{Experiments}
\paragraph{Motivation.} We run our algorithms on real-life data and on synthetic data with realistic structure. We draw two main conclusions from these experiments.\footnote{All code and data is available at https://github.com/martijnmartijnmartijnmartijn/Diverse-Committees-with-Incomplete-or-Inaccurate-Approval-Ballots.} 

First, we establish empirically that the greedy and local search algorithms with complete information (Algorithm~\ref{alg:greedy} and \ref{alg:LS}) consistently achieve a higher CC-score than two well known committee election algorithms, i.e. \textsc{approval voting} (selecting the $k$ most popular candidates) and the \textsc{local search pav} algorithm of \citet{halpern_representation_2023}. We take this to justify the study of our querying algorithms, as opposed to querying versions of these other algorithms, since the querying algorithms approximate the performance of the complete information algorithms. 

Second, although we have proven that our querying algorithms can overcome inaccurate and incomplete voter responses with constant overhead, this constant overhead can still be very large. We show that in practice our algorithms beat the optimal approximation ratio, even when using much fewer queries than theoretically required, and even when we combine inaccurate and incomplete information (a setting that we did not study theoretically).

\paragraph{Data.}
We test our algorithms on 18 open-use datasets of Polis deliberations.\footnote{https://github.com/compdemocracy/openData.} As anticipated, participants typically vote on only a small portion of the statements, so we obtain a sparse comment-participant-matrix. Since these discussions have been completed already, we are not able to query voters anymore, so, for the sake of running our querying algorithms, we need to artificially complete the data. For details on this and other pre-processing steps we implemented, we refer to the full version of the paper.

As the 18 data sets available are too few to support robust observations,
we generate synthetic data that are structurally similar to the Polis data. 
To do so, we turn to established methods for sampling approval-based elections, introduced by \citet{szufa_how_2022}.
Specifically, we employ the $(q,\phi)$-resampling model to sample approval elections which, among all models discussed by \citet{szufa_how_2022}, provided the best fit with respect to the 18 data sets available.
In this model, $q\in[0,1]$ represents the fraction of approvals and $\phi\in[0,1]$ represents the spread of approvals, so that $\phi=0$ means all voters are identical and approve the exact same $\lfloor q\cdot m\rfloor$ candidates, while $\phi=1$ means each candidate is approved with probability $q$ independently so that the spread of approvals is maximal. After pre-processing the Polis data, we find $q=0.0891$ and $\phi=0.693$. The full version of the paper contains an explanation of how to arrive at these values. We sampled 100 elections according to the $(0.0891, 0.693)$-resampling model with $n=1000$ and $m=400$, as $m/n=0.4$ on average for the Polis datasets, and the average number of voters is roughly $1000$.

\paragraph{Conclusion 1: complete information algorithms}
\begin{figure*}[t]
\includegraphics{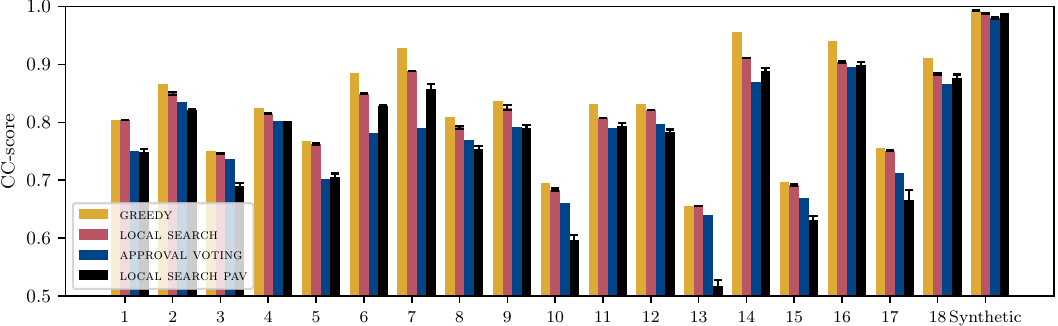}
    \caption{CC-score achieved by algorithms \textsc{greedy}, \textsc{local search} (Algorithms~\ref{alg:greedy} and~\ref{alg:LS}), \textsc{approval voting} and \textsc{local search pav} on the 18 Polis data sets (plotted separately) and 100 synthetic data sets (plotted together) for $k=8$. For the two local search algorithms, we ran 20 random trials (due to the random initial committee), plotting the mean and standard deviation. We used Paul Tol's high contrast colour scheme, designed to be color blind safe. See the full version of the paper for the names corresponding to the numbers of the datasets.}
    \Description{CC-score achieved by algorithms `greedy', `local search', `approval voting' and `local search pav' on the 18 Polis data sets (plotted separately) and 100 synthetic data sets (plotted together) for $k=8$. For the two local search algorithms, we ran 20 random trials (due to the random initial committee), plotting the mean and standard deviation.}
    \label{fig:plot2}
\end{figure*}

Our first question was whether \textsc{greedy} and \textsc{local search}-$\beta$ achieve better CC-score in practice compared to two well-known committee election algorithms: \textsc{approval voting} and \textsc{local search pav}.\footnote{We configure \textsc{local search pav} with $\alpha=1$ which is the best performing configuration retaining a provably polynomial time runtime \cite{halpern_representation_2023}.} To answer this question, we ran \textsc{greedy} (Algorithm~\ref{alg:greedy}), \textsc{local search}-$\beta$ (Algorithm~\ref{alg:LS})\footnote{We took $\beta$ as in Theorem~\ref{thm:ls_bound} and write \textsc{local search} going forward.}, \textsc{approval voting} and \textsc{local search pav} on the 118 data sets, running 20 random trials per dataset for both local search algorithms because of the random starting committee. 

Figure~\ref{fig:plot2} shows the CC-score attained by the four complete information algorithms on the 118 datasets. Since the synthetic data are drawn from the same distribution, we show the mean CC-score. This is not the case for the Polis data, which differ quite significantly, so that we plot the CC-score on each Polis dataset separately.

We can see that our \textsc{greedy} and \textsc{local search} algorithms achieve a higher (mean) CC-score than both \textsc{approval voting} and \textsc{local search pav} across all Polis datasets, as well as on the synthetic data. Averaged across the Polis datasets, the best performing of our two algorithms achieves a CC-score $8.6\%$ and $6.3\%$ higher than \textsc{local search pav}, \textsc{approval voting}, respectively. Across the synthetic data, this is $0.67\%$, $1.3\%$, respectively. We note that the CC-score on the synthetic data is generally very close to 1, leaving less room for improvement to begin with. At times \textsc{approval voting} outperforms \textsc{local search pav}, which may be unexpected (e.g., datasets 3, 10, 15). We conclude from this that our two algorithms are significant improvements over these existing algorithms when the objective is to select diverse committees.



\paragraph{Conclusion 2: querying algorithms}
\begin{figure*}[t]
\includegraphics{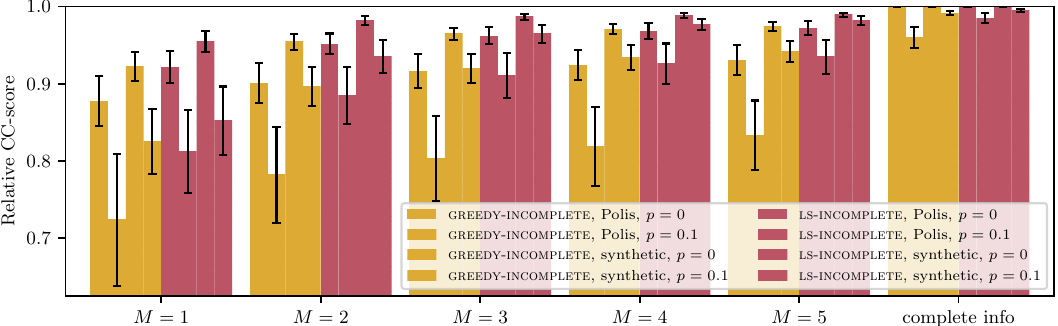}
    \caption{CC-score of \textsc{greedy-incomplete} (Algorithm~\ref{alg:greedy-incomplete-queries}) and \textsc{ls-incomplete} (Algorithm~\ref{alg:LS-incomplete-queries}) on the 18 Polis and 100 synthetic data sets for $k=8$. The algorithms are configured so that each voter is expected to be queried only $M=1,2,3,4,5$ times. This was done for $p=0$ (accurate responses) and $p=0.1$ (inaccurate responses). On each of the 118 data sets, we ran 50 random trials of the algorithms, plotting the mean and standard deviation. All scores shown are relative to the complete and accurate information setting. We used Paul Tol's high contrast colour scheme, designed to be color blind safe.}
    \Description{CC-score of `greedy-incomplete' and `ls-incomplete' on the 18 Polis and 100 synthetic data sets for $k=8$. The algorithms are configured so that each voter is expected to be queried only $M=1,2,3,4,5$ times. This was done for $p=0$ (accurate responses) and $p=0.1$ (inaccurate responses). On each of the 118 data sets, we ran 50 random trials of the algorithms, plotting the mean and standard deviation. All scores shown are relative to the complete and accurate information setting.}
    \label{fig:plot}
\end{figure*}


To answer our second question, we take the CC-scores of \textsc{greedy} and \textsc{local search} from Figure~\ref{fig:plot2} as our baseline and inspect how close we can get to these scores when information is incomplete and/or inaccurate, using a realistic number of queries. For the incomplete information model, taking $t=20$, we write $M$ for the expected number of query sets of size $t$ that each voter is presented with. We then configure \textsc{greedy-incomplete} (Algorithm~\ref{alg:greedy-incomplete-queries}) and \textsc{ls-incomplete} (Algorithm~\ref{alg:LS-incomplete-queries}) so that $M=1,2,3,4,5$, which is much lower than what would theoretically be required. We do 50 random trials of each algorithm (due to the random querying of voters). For the inaccurate information model, we ignore the theoretically required repetitions of the queries, and run the standard greedy and local search algorithms, but with $p=0.1$\footnote{We ran the experiments for multiple values of $p$ and the results were consistent. We also ran the experiments for different values of $k$ and the results  were similar.}. Finally, we combine both of the above by running Algorithm~\ref{alg:greedy-incomplete-queries} and Algorithm~\ref{alg:LS-incomplete-queries} with the above values of $M$ but for $p=0.1$.

Figure~\ref{fig:plot} shows the {CC-scores obtained in these experiments. We see that both algorithms with accurate responses ($p=0$) obtain a score very close (0.85--0.95) to versions of the algorithms with complete information, even when $M=1$. With complete information, both algorithms obtained an average {CC-score of around $0.8$. Multiplying by $0.9$ still yields a score above the worst-case approximation ratio of $1-1/\e\approx 0.63$, even when the optimal solution would have a CC-score of 1. Example~\ref{ex:experiments1} highlights the great performance witnessed in the experiments for $p=0$. This shows that even with limited querying of voters, we can reliably attain a diverse committee.

\begin{example}\label{ex:experiments1}
Figure~\ref{fig:plot} shows that for $M=5$, the CC-score attained by \textsc{ls-incomplete}, on average across 18 Polis datasets with 50 runs per set, is approximately 0.95 times the score attained by \textsc{local search}. The dataset \texttt{vtaiwan.uberx} has $m=197$ and $n=1921$. For these values of $m$ and $n$, taking $\gamma=0.95$, $k=8$, $t=20$, and $\delta=0.05$, we would need $M=7.109\cdot 10^8$ to obtain the guaranteed ratio of $(1-1/\e)\cdot 0.95$ with probability $0.95$, using the upper bound proven in Theorem~\ref{thm:ls_bound}.
\end{example}

With $p=0.1$, but with complete information, the performance decreases by a factor of 0.95--0.97 on average, which yields a score confidently above the worst-case approximation ratio of $1-1/\e$. As a comparison: taking again dataset \texttt{vtaiwan.uberx} with $\delta=0.05$ would require repeating each query 32 times to obtain the guaranteed ratio of $(1-1/\e)$ with probability $0.95$, according to the upper bound in Theorem~\ref{thm:inaccurate_upper_bound}.

Taking both $p=0.1$ \emph{and} incomplete information, 
not meeting the theoretically required number of queries for \emph{either} scenario, the performance takes a noticeable hit. The algorithms perform worse than the versions with complete information by a factor of 0.7--0.8 for $M=1$ increasing to 0.85--0.95 for $M=5$. However, this still implies that we are above the worst-case approximation ratio. Do note that the standard deviation can become relatively large here, especially for the Polis data, which appear to be less homogeneous than the synthetic data. 

\section{Outlook}
Our work is the first to address diversity in approval-based committee elections in the context of online civic participation platforms. Measuring diversity by the Chamberlin-Courant score, we proved diverse committees can be found by querying only a small fraction of the voters, even when responses may be inaccurate. This remains true in the presence of external diversity constraints on the committee, such as quota. Our algorithms match lower bounds on the query complexity (up to log-factors). We verify these theoretical results empirically on both real-life and synthetic data.

Our results open up several directions for future research. 
{\em First}, it would be interesting to combine the incomplete and inaccurate information models explicitly in theoretical analysis.
{\em Second}, it would be desirable to lift our results on inaccurate information to richer error models, e.g., where the error probability changes over time and/or between voters. 
{\em Third}, one could study more adaptive algorithms that can decide whom to best query about which comment at which time---the so-called `comment routing problem' \cite{small2021polis}. In doing so, the query complexity can perhaps be lowered further. {\em Fourth}, one can explore the significance of our results for active learning \cite{settles2009active}, with the aim of optimizing the querying of annotators in distributed data labelling tasks.

\begin{acks}
We want to thank the anonymous reviewers of EC'25, COMSOC'25 and AAAI'26 for their helpful comments and suggestions. Feline, Davide and Pradeep were supported by the Hybrid Intelligence Center, a 10-year programme funded by the Dutch Ministry of Education, Culture and Science through the Netherlands Organisation for Scientific Research, \url{https://www.hybrid-intelligence-centre.nl/}, under Grant No. (024.004.022). Davide ackowledges support by the European Union under the Horizon Europe project Perycles (Participatory Democracy that Scales, \url{https://perycles-project.eu/)}.

This project was funded by the European Union. Views and opinions expressed are however those of the author(s) only and do not necessarily reflect those of the European Union or the European Research Executive Agency. Neither the European Union nor the granting authority can be held responsible
for them.
\smallskip
\begin{center}
\includegraphics[width=0.22\textwidth]{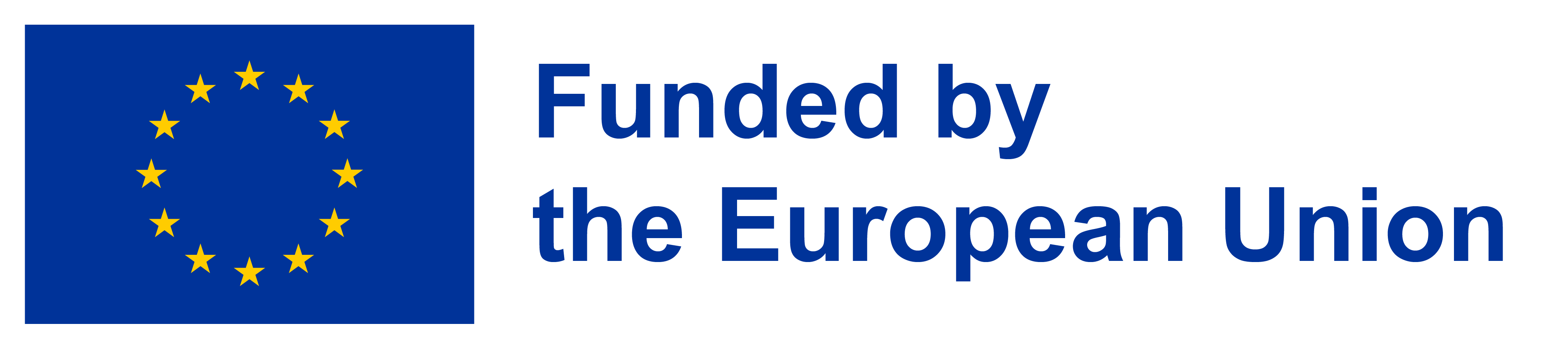}
\end{center}
\end{acks}

\balance 
\bibliographystyle{ACM-Reference-Format}
\bibliography{BIEB}

\onecolumn 
\appendix
\section{Proof of Theorem~\ref{thm:greedy_bound}}\label{app:Greedy}
We prove that \textsc{greedy-incomplete} (Algorithm~\ref{alg:greedy-incomplete-queries}) can get arbitrarily close to the optimal approximation ratio for CC using $O(m \log m)$ queries to voters. This result was stated as Theorem~\ref{thm:greedy_bound}. We start by giving a new proof that the regular greedy algorithm (Algorithm \ref{alg:greedy}) achieves the optimal approximation ratio of $1-\frac{1}{\e}$, and we adapt this proof for Lemma~\ref{lem:greedy_eps}, and consequently for the main result, Theorem~\ref{thm:greedy_bound}.
\begin{lemma}\label{lem:greedy_cc_ratio}
Algorithm~\ref{alg:greedy} is $(1-\frac{1}{\e})$-approximate for CC, and runs in time $O(kmn)$.
\end{lemma}
For the proof we use the following notation: for $i\le k$, let $\gr_i$ denote the increase in CC-score the greedy algorithm attains in iteration $i$, and $\gr(i)$ for the CC-score of the greedy solution after $i$ iterations. This means $\gr(k)$ is the score after termination. We write $\opt(k)$ for the score an optimal solution, consisting of $k$ candidates, would have.
\begin{proof}
We start with the approximation ratio. Like \cite[Lemma 1]{hochbaum_analysis_1998}, we first prove that, for $\ell\le k$, $\gr_{\ell}\ge (\opt(k)-\gr(\ell-1))/k$. Fix an optimal solution. At the start of iteration $\ell$, the difference between the greedy score and the optimal score is $\opt(k)-\gr(\ell-1)$. Thus, there are \emph{at least} this many elements covered in the (fixed) optimal solution, that have not been covered yet by \textsc{greedy} up to that point (it could be more: the set selected by \textsc{greedy} at the start of iteration $\ell$, needn't be a subset of the optimal solution). These $\opt(k)-\gr(\ell-1)$ elements are covered by the $k$ sets in the optimal solution so, by the pigeon hole principle, there exists a set in the optimal solution, that is not currently selected by \textsc{greedy}, that will cover at least $(\opt(k)-\gr(\ell-1))/k$ of these elements. Then, in iteration $\ell$, the greedy algorithm picks the set that covers the most additional elements, so these must be at least $(\opt(k)-\gr(\ell-1))/k$ many. We use this relation to obtain an expression for $\gr(k)$ in terms of $\opt(k)$. For all $k$,

\begin{align*}
\gr(k)&=\gr(k-1)+\gr_k\\
&\ge \gr(k-1)+\frac{\opt(k)-\gr(k-1)}{k}\\
&=\frac{\opt(k)}{k}+\gr(k-1)\left(1-\frac{1}{k}\right)\\
&\ge \frac{\opt(k)}{k}+\left(\frac{\opt(k)}{k}+\gr(k-2)\left(1-\frac{1}{k}\right)\right)\left(1-\frac{1}{k}\right)\\
&\ge \frac{\opt(k)}{k}+\left(\frac{\opt(k)}{k}+\left(\frac{\opt(k)}{k}+\gr(k-3)\left(1-\frac{1}{k}\right)\right)\left(1-\frac{1}{k}\right)\right)\left(1-\frac{1}{k}\right) \\ 
&\ldots \\
&\ge \frac{\opt(k)}{k}+\frac{\opt(k)}{k}\left(1-\frac{1}{k}\right)+\frac{\opt(k)}{k}\left(1-\frac{1}{k}\right)^2+\ldots+\frac{\opt(k)}{k}\left(1-\frac{1}{k}\right)^{k-2}+\gr(1)\left(1-\frac{1}{k}\right)^{k-1}\\
&\ge \frac{\opt(k)}{k}\left(\sum_{i=0}^{k-2}\left(1-\frac{1}{k}\right)^i+\left(1-\frac{1}{k}\right)^{k-1}\right)\\
&=\frac{\opt(k)}{k}\sum_{i=0}^{k-1}\left(1-\frac{1}{k}\right)^i\\
&= \frac{\opt(k)}{k}\left(k\left(1-\left(1-\frac{1}{k}\right)^k\right)\right)\\
&= \opt(k)\left(1-\left(1-\frac{1}{k}\right)^k\right).
\end{align*}
Because $\left(1-\frac{1}{k}\right)^k$ is an increasing sequence with $\lim_{k\to\infty}\left(1-\frac{1}{k}\right)^k=\frac{1}{\e}$, we can bound the last expression from below by $\opt(k)\left(1-\frac{1}{\e}\right)$, which proves the approximation ratio. 

Regarding the time complexity, the algorithm chooses $k$ candidates sequentially. For each of these $k$ steps, we have to go over all $m$ candidates to check which adds most to the current score. Per candidate, we have to check for each of the $n$ voters whether it is already represented (we can keep track of the represented voters in one set). This gives a time complexity in $O(k m n)$.
\end{proof}

\begin{remark}
The convergence goes rather quickly; for $k=10$, we already have $1-(1-\frac{1}{k})^k=1-\frac{9}{10}^{10}\approx 0.651$, whereas $1-1/\e\approx 0.632$. However, for small committee sizes the actual worst case approximation ratio can be substantially smaller than the general ratio of $1-1/\e$: when $k=3$, the ratio is actually $1-\frac{2}{3}^3\approx 0.704$.
\end{remark}

\begin{proposition}\label{prop:ratio_greedy_eps}
When $\sum_{v_i\in V}\mathds{1}_{\{|A(i)|\ge 1\}}(i)\ge k$, then $\gr(k)\ge k$. 
\end{proposition}
\begin{proof}
We proceed by contraposition. Assume $\gr(k)<k$, which is equivalent to $\gr(k)\le k-1$, so (by the pigeonhole principle) there exists $i\le k$ such that $\gr_i=0$. Because $\gr_i \le \gr_{i-1}$ for all $i\le k$, we must have (at least) $\gr_k=0$; no candidate with a yet uncovered voter is left in iteration $k$, since otherwise the greedy algorithm would pick that over an empty candidate and then $\gr_k>0$. At this point, $\gr(k-1)\le k-1$ and no uncovered voters are left, so $\sum_{v_i\in V}\mathds{1}_{\{|A(i)|\ge 1\}}(i)\le k-1$, which means the instance is trivial.
\end{proof}

For the incomplete information setting we need to estimate approval information of the total population based on the information of only few voters. This can lead to errors. Write $\gr\hspace{-0.5mm}\text{-}\epsilon$ for any adaptation of Algorithm~\ref{alg:greedy} that, in each iteration, chooses any candidate that yields the largest increase in score \emph{up to $-\epsilon$}; see Algorithm~\ref{alg:greedy-eps} for an example. We obtain the following approximation.
\begin{lemma}\label{lem:greedy_eps}
When $\sum_{v_i\in V}\mathds{1}_{\{|A(i)|\ge 1\}}(i)\ge k$, any \textsc{greedy}-$\epsilon$ algorithm is $\left(\left(1-\frac{1}{\e}\right)\cdot\left(\frac{1}{1+\epsilon(1-\frac{1}{\e})}\right)\right)$-approximate.
\end{lemma}
For $i\le k$, define $\gr\hspace{-0.5mm}\text{-}\epsilon_i$ and $\gr\hspace{-0.5mm}\text{-}\epsilon(i)$ similarly to before.
\begin{proof}
We write out $\gr\hspace{-0.5mm}\text{-}\epsilon(k)$ recursively like before, but now subtracting $\epsilon$ in each iteration. For any $k$, we have
\begin{align*}
\gr\hspace{-0.5mm}\text{-}\epsilon(k)&=\gr\hspace{-0.5mm}\text{-}\epsilon(k-1)+\gr\hspace{-0.5mm}\text{-}\epsilon_k\\
&\ge \gr\hspace{-0.5mm}\text{-}\epsilon(k-1)+\frac{\opt(k)-\gr\hspace{-0.5mm}\text{-}\epsilon(k-1)}{k}-\epsilon\\
&=\frac{\opt(k)}{k}+\gr\hspace{-0.5mm}\text{-}\epsilon(k-1)\left(1-\frac{1}{k}\right)-\epsilon\\
&\ge \frac{\opt(k)}{k}+\left(\frac{\opt(k)}{k}+\gr\hspace{-0.5mm}\text{-}\epsilon(k-2)\left(1-\frac{1}{k}\right)-\epsilon\right)\left(1-\frac{1}{k}\right)-\epsilon\\
&= \frac{\opt(k)}{k}-\epsilon + \left(\frac{\opt(k)}{k}-\epsilon\right)\left(1-\frac{1}{k}\right)+\gr\hspace{-0.5mm}\text{-}\epsilon(k-2)\left(1-\frac{1}{k}\right)^2\\
&\ge \frac{\opt(k)}{k}-\epsilon +\left(\frac{\opt(k)}{k}-\epsilon\right)\left(1-\frac{1}{k}\right)+\left(\frac{\opt(k)}{k}-\epsilon\right)\left(1-\frac{1}{k}\right)^2+\gr\hspace{-0.5mm}\text{-}\epsilon(k-3)\left(1-\frac{1}{k}\right)^3
&\ldots \\
&\ge \frac{\opt(k)}{k}-\epsilon +\left(\frac{\opt(k)}{k}-\epsilon\right)\left(1-\frac{1}{k}\right)+\left(\frac{\opt(k)}{k}-\epsilon\right)\left(1-\frac{1}{k}\right)^2\\
&\ \ \ +\ldots+ \left(\frac{\opt(k)}{k}-\epsilon\right)\left(1-\frac{1}{k}\right)^{k-2}+\gr\hspace{-0.5mm}\text{-}\epsilon(1)\left(1-\frac{1}{k}\right)^{k-1}\\
&\ge \left(\frac{\opt(k)}{k}-\epsilon\right)\sum_{i=0}^{k-1}\left(1-\frac{1}{k}\right)^i\\
&= \opt(k)\left(1-\left(1-\frac{1}{k}\right)^k\right) -\epsilon k\left(1-\left(1-\frac{1}{k}\right)^k\right),
\end{align*}
which we can bound from below by $(\opt(k)-\epsilon k)\left(1-\frac{1}{\e}\right)$ again by taking a limit $k\to\infty$. This shows $\frac{\gr\hspace{-0.5mm}\text{-}\epsilon(k)}{\opt(k)}\ge \left(1-\frac{\epsilon k}{\opt(k)}\right)\left(1-\frac{1}{\e}\right) = \left(1-\frac{1}{\e}\right)+\frac{1}{\opt(k)}\left(\frac{\epsilon k}{\e}-\epsilon k\right)\iff \frac{\gr\hspace{-0.5mm}\text{-}\epsilon(k) +\epsilon k\left(1-\frac{1}{\e}\right)}{\opt(k)}\ge 1-\frac{1}{\e}.$ Using $\gr(k)\ge k$ (Proposition~\ref{prop:ratio_greedy_eps}), we can reformulate $1-\frac{1}{\e}\le \frac{\gr\hspace{-0.5mm}\text{-}\epsilon(k) +\epsilon k\left(1-\frac{1}{\e}\right)}{\opt(k)}\le \frac{\gr\hspace{-0.5mm}\text{-}\epsilon(k) +\gr\hspace{-0.5mm}\text{-}\epsilon(k)\cdot\epsilon\left(1-\frac{1}{\e}\right)}{\opt(k)}= \frac{\gr\hspace{-0.5mm}\text{-}\epsilon(k)\cdot\left(1+\epsilon\left(1-\frac{1}{\e}\right)\right)}{\opt(k)}$.
\end{proof}
\begin{algorithm}[h]
	\SetAlgoNoLine
	\KwIn{Numbers $n,m,k\in\N$ with $k\le m$, set $V$ of $n$ voters, set $C$ of $m$ candidates.}
	\KwOut{Committee $W\subseteq C$ of size $k$.}
    Let $W=\{\}$ be an empty set\;
    \For{$i=1,\ldots,k$,}{
		\begin{enumerate}
			\item  For all $c\notin W$, determine $\Delta(W,c)$.
			\item Add to $W$ any $c'\notin W$ for which $\Delta(W,c')\in [\max_{c\notin W}\,\Delta(W,c)-\epsilon,\max_{c\notin W}\,\Delta(W,c)$], chosen uniformly at random.
		\end{enumerate}
	}
    \caption{A \textsc{greedy-$\epsilon$} algorithm}
    \label{alg:greedy-eps}
\end{algorithm}

\begin{reptheorem}{thm:greedy_bound}
Let $\sum_{v_i\in V}\mathds{1}_{\{|A(i)|\ge 1\}}(i)\ge k$, $\delta>0,\, \gamma\in(0,1)$, and $k<t\le m$. Then, w.p. at least $1-\delta$, Algorithm~\ref{alg:greedy-incomplete-queries} is $(1-1/e)\gamma$-approximate for CC with query complexity  
$O\left(\left(\frac{\gamma}{1-\gamma}\right)^2km\log\left(\frac{km}{\delta}\right)\right)\in O_{\delta,\gamma,k}(m\log m)$.
\end{reptheorem}

\begin{proof}
We show that the greedy query algorithm, in order to be $(1-1/\e)\gamma$-approximate for CC with stated probability and number of queries, indeed needs the values of $\epsilon$ and $\ell$ as stated in Algorithm~\ref{alg:greedy-incomplete-queries}. 

For $\gamma\in(0,1)$, we take $\epsilon=\frac{(1-\gamma)e}{\gamma(e-1)}$ so that $\frac{\gr\hspace{-0.5mm}\text{-}\epsilon(k)}{\opt(k)}\ge\gamma\cdot\left(1-\frac{1}{\e}\right)$, using Lemma~\ref{lem:greedy_eps}. We want to find $\ell$ so that with probability $1-\delta$, all our estimations $\hat{\Delta}$ are within  $\frac{1}{2}\epsilon$ of the true values $\Delta$. For this, we use Hoeffding's inequality to bound the probability that a single estimate $\hat{\Delta}$ is more than $\frac{1}{2}\epsilon$ away from its corresponding true value $\Delta$. Write $X_{W,c}:V\to\{0,1\}:v_i\mapsto\mathds{1}_{\{A(i)\cap (W\cup\{c\})\neq\emptyset\}}(i)-\mathds{1}_{\{A(i)\cap W\neq\emptyset\}}(i)$ for a random variable with expected value $\mathbb{E}[X_{W,c}]=\Delta(W,c)$. We can draw from $X_{W,c}$ by querying individual voters about $W\cup\{c\}$, then $\ell$ independent realizations with values in $\{0,1\}$ yield sample mean $\hat{\Delta}(W,c)$, with $\mathbb{E}[\hat{\Delta(W,c)}]=\Delta(W,c)$. Hoeffding's inequality gives, for any estimate $\hat{\Delta}$ of $\Delta$:
\begin{align}
\P\left(|\hat{\Delta}-\Delta|\ge\frac{1}{2}\epsilon\right)\le 2\exp\left(-\frac{1}{2}\ell\epsilon^2\right).
\end{align}
There are $k$ iterations and at most $m$ estimates $\hat{\Delta(W,c)}$ per iteration so no more than $k\cdot m$ estimations in total. Thus, by independence, we have
\begin{align}
\P\left(\forall i: |\hat{\Delta}_i-\Delta_i|<\frac{1}{2}\epsilon\right)\ge 1-2\exp\left(-\frac{1}{2}\ell\epsilon^2\right)\cdot  m \cdot k.
\end{align}
We want to find $\ell$ so that this value equals at least $1-\delta$, and so we proceed to write
\begin{align}
1-2\exp\left(-\frac{1}{2}\ell\epsilon^2\right)\cdot m \cdot k\ge1-\delta 
\iff \exp\left(-\frac{1}{2}\ell\epsilon^2\right)\le\frac{\delta}{2 m k}
\iff \ell\ge \frac{2}{\epsilon^2}\log\left(\frac{2 m k}{\delta}\right).
\end{align}
This means, for any $\delta>0$, any $\ell$ that is at least this large, will, with probability $1-\delta$, produce (querying $\ell$ voters) estimates that are all at most $\frac{1}{2}\epsilon$ apart from the value that they estimate. Then we can invoke Lemma~\ref{lem:greedy_eps} to obtain the desired ratio. We take the smallest possible value: $\ell = \left\lceil \frac{2}{\epsilon^2}\log\left(\frac{2 m k}{\delta}\right)\right\rceil$ and filling in for $\epsilon$ this results in
\begin{align}
\ell=\left \lceil \frac{2}{\left(\frac{(1-\gamma)e}{\gamma(e-1)}\right)^2}\log\left(\frac{2 m k}{\delta}\right)\right \rceil = \left \lceil \frac{2\gamma^2(e-1)^2}{(1-\gamma)^2\e^2}\log\left(\frac{2 m k}{\delta}\right)\right \rceil \le \left \lceil 0.8\left(\frac{\gamma}{1-\gamma}\right)^2\log\left(\frac{2 m k}{\delta}\right)\right \rceil \
\end{align}
There are $k$ iterations and $t\left\lceil\frac{m-k}{t-k}\right\rceil\ell$ queries per iteration, so the total number of queries to $A$ is
\begin{align}
k\cdot t \left\lceil\frac{m-k}{t-k}\right\rceil\cdot \ell\le k\cdot t \left\lceil\frac{m-k}{t-k}\right\rceil\cdot\left \lceil 0.8\left(\frac{\gamma}{1-\gamma}\right)^2\log\left(\frac{2 m k}{\delta}\right)\right \rceil \in O\left(\left(\frac{\gamma}{1-\gamma}\right)^2km\log\left (\frac{km}{\delta}\right)\right).
\end{align}
\end{proof}

\section{Matroid for quota}\label{app:matroid}
\begin{example}[matroid implementing quota, example from \cite{masarik_generalised_2024}]\label{ex:quota}
Suppose $C=\cup_{i=1}^s C_i$ with $C_i\cap C_j =\emptyset$ for all $i,j\in [s], i\neq j$. For each group $i$ we have an upper and a lower quota, denoted by $q_i^{\top}$ and $q_i^{\bot}$ respectively. Define 
\[
I=\{W\subseteq C:|W|=k \text{ and } q_i^{\bot}\leq |W\cap C_i|\leq q_i^{\top} \text{ for all }i\in [s]\}
\]
and $\mathcal{I}=I \bigcup \ \{W':W'\subseteq W,\ W\in I\}$. Then $\mathcal{M}=(C,\mathcal{I})$ is a matroid of rank $k$. For a proof, see \cite{masarik_generalised_2024}.
\end{example}
In Example~\ref{ex:quota}, only the bases comply with the quota, but this is not a problem. When an algorithm outputs $A\in\mathcal{I}$ with $|A|=\ell<k$, then by requirement 3) in Definition~\ref{matroid}, using any basis element $B=\{b_1,\ldots,b_k\}\in\mathcal{I}$, we can build a path $A,\ A\cup \{b_1\},\ A\cup \{b_1\}\cup \{b_2\},\ \ldots, A\cup \{b_1\}\cup\ldots\cup\{b_{k-\ell}\}=: A'$ from $A$ to attain a basis element $A'\in\mathcal{I}$ that contains $A$ (and has a score at least as high). Moreover, our local search algorithm only outputs solutions of size $k$ so this is never a problem to begin with.

\section{Proof of Theorem~\ref{thm:ls_bound}}\label{app:LS}
We prove that \textsc{ls-$\beta$-incomplete} (Algorithm~\ref{alg:LS-incomplete-queries}) can get arbitrarily close to the optimal approximation ratio for CC using $O(m \log m)$ queries to voters. This result was stated as Theorem~\ref{thm:ls_bound}. First, we discuss some changes we made to the complete local search algorithm of \cite{filmus_power_2013} (resulting in Algorithm~\ref{alg:LS}), as these changes shine through in the incomplete local search algorithm.

The first change we made is that we have left out the greedy initialization. Although it speeds up the convergence time-wise, it would cost many extra querying rounds. It is not possible to query only for the initialization and use the gathered information for the local search phase, since candidates not elected in the initialization, or only elected in its last round, will not have query responses from all voters whereas that is necessary if the candidate would become part of the tentative solution at some point during the local search. However, it might be possible to save the information gathered in the initialization in a smart way. For now, we leave out this phase. For similar reasons, we omit the partial enumeration step that would yield the clean $1-1/\e$ ratio. Secondly, we have adapted the stopping condition from $\frac{f((W\cup c')\setminus c)}{f(W)}\le 1+\beta n$ (the factor $n$ is added because our scores range from $[0,1]$ instead of $[0,n]$) to $f((W\cup c')\setminus c)-f(W)\le\beta$. Before we can make any kind of statement about the performance of our local search querying algorithm, we need guarantees on the accuracy of our estimations $\hat{\Delta}$. We use Hoeffding's inequality to bound the probability that our estimates $\hat{\Delta}$ are far from the estimated true value $\Delta$, but for that we need our samples to be bounded as well. For our stopping condition $\hat{\Delta}(W,c',c)=\hat{f}(W\cup\{c'\}\setminus\{c\})-\hat{f}(W)<\beta-\epsilon$, these samples are really just any difference between two consecutive values of $(\alpha_i)_{i\in \mathbb{N}}$, since replacing $c$ by $c'$, any voter can be represented one time less, the same number of times, or one time more. The samples are thus bounded from above by $\alpha_1-\alpha_0=\alpha_1$. However, since $\alpha_0=0$, samples for the stopping condition $\frac{f((W\cup c')\setminus c)}{f(W)}\le 1+\beta n$ are unbounded. Luckily, when $f((W\cup c')\setminus c)-f(W)\le\beta$, then also $\frac{f((W\cup c')\setminus c)}{f(W)}\le 1+\frac{\beta}{f(W)}$, which is at most $1+\beta n$ iff $f(W)\ge 1/n$. This will always be true as long as $\forall c_j\ \exists i: c_j\in A(i)$; all candidates have at least one approval, and $k\ge 3$, so that for any elected committee  $W$, $f(W)\ge \alpha_3/n>1/n$.
\begin{reptheorem}{thm:ls_bound}
Let $|c_j|\ge 1\ \forall j$, $\delta>0$, $\gamma\in(0,1)$, $m\ge t > k\ge 3$. Fix $\beta = C_2\frac{1-\gamma}{\gamma k\log k}$ for some constant $C_2$. Then, w.p. at least $1-\delta$, Algorithm~\ref{alg:LS-incomplete-queries} is $(1-1/e)\gamma$-approximate for CC with query complexity
$ O\left(\left(\frac{\gamma k\log k}{1-\gamma}\right)^3m\log\left(\frac{m\gamma k^2\log k}{\delta(1-\gamma)}\right)\right)\in O_{\delta,\gamma,k}(m\log m)$.
\end{reptheorem}

\begin{proof}
We show that the local search query algorithm, in order to be $(1-1/\e)\gamma)$-approximate for CC with stated probability and number of queries, indeed needs the values of $\epsilon$ and $\ell$ as in Algorithm~\ref{alg:LS-incomplete-queries}. In the original approximation ratio (of Algorithm~\ref{alg:LS}), we substitute $\gamma$ for $(1-1/\e)(1-\gamma)$, so that the ratio becomes $(1-1/\e)\cdot\gamma$ and we adapt the constant factor accordingly. 

With $\beta=C_2\frac{1-\gamma}{\gamma k\log k}$, the minimum step size of \textsc{local search}-$\beta$ (Algorithm~\ref{alg:LS}) is $C_2\frac{1-\gamma}{\gamma k\log k}$, and every committee attains score at least $0$ and at most $\alpha_k$, so Algorithm~\ref{alg:LS} runs at most $\alpha_k \cdot \frac{\gamma k\log k}{C_2(1-\gamma)}$ iterations. We anticipate that querying will increase the number of required iterations, because of the uncertainty inducing estimations. Depending on how many iterations we allow, values of $\epsilon$ (the estimation margin) and $\ell$ (the number of voters queried per round) will also differ. Thus, to make this dependency explicit, we add in a factor $\xi\ge 1$ and let Algorithm~\ref{alg:LS-incomplete-queries} have a total of at most $\xi \alpha_k \cdot \frac{\gamma k\log k}{C_2(1-\gamma)}$ iterations. We then compute the maximal value of $\epsilon$ this corresponds to. 

Until the last run of the algorithm, we have $\hat{\Delta}(W,c',c)>\beta-\epsilon$, and when this is no longer true, the swap of $c$ for $c'$ is not made. We take $\epsilon$ so that all true improvements ${\Delta}(W,c',c)$ are at least size $\hat{\Delta}(W,c',c)-\epsilon >\beta-2\epsilon$ and for that, we solve
\begin{align}
\beta-2\epsilon= C_2\frac{1-\gamma}{\gamma k\log k}-2\epsilon = C_2\frac{1-\gamma}{\xi\cdot \gamma k\log k},
\end{align}
for $\epsilon$ to obtain $\epsilon=\frac{\xi - 1}{2\xi}\cdot C_2\frac{1-\gamma}{\gamma k\log k}$, which is the maximum value of $\epsilon$ we may allow in order to have our original number of iterations (for Algorithm~\ref{alg:LS}) multiplied by at most $\xi$.

We take $\ell$ such that with probability $1-\delta$, all our estimations $\hat{\Delta}$ are within $\epsilon$ of the true values $\Delta$. For this, we use Hoeffding's inequality to bound the probability that a single estimate $\hat{\Delta}$ is more than $\epsilon$ away from its corresponding true value $\Delta$. Write $X_{W,c',c}:V\to [-\alpha_1,\alpha_1]:v_i\mapsto\alpha_{h_i((W\cup\{c'\})\setminus\{c\})}-\alpha_{h_i(W)}$ for a random variable with expected value $\mathbb{E}[X_{W,c',c}]=\Delta(W,c',c)$. We can draw from $X_{W,c,c'}$, indirectly, by querying individual voters about $(W\cup\{c'\})\setminus\{c\}$, then $\ell$ independent realizations yield sample mean $\hat{\Delta}(W,c',c)$, with $\mathbb{E}[\hat{\Delta(W,c',c)}]=\Delta(W,c',c)$. Regarding the range of $X_{W,c',c}$: For any committee $W$, any voter $v_i$ will attain some score $\alpha_{h_i(W)}\in\{\alpha_0,\ldots,\alpha_k\}$ and whenever a candidate $c\in W$ is replaced for another candidate $c'\notin W$, the voter can approve one less candidate than before, the same number of candidates as before, or one more candidate than before. A score $\alpha_i$ before a swap, can thus become any of scores $\{\alpha_{i-1},\alpha_i,\alpha_{i+1}\}$ after. In Lemma 3, \cite{filmus_power_2013} prove that for all $i<k$, $\alpha_{i+1}>\alpha_i$ and $\alpha_{i+2}-\alpha_{i+1}\le\alpha_{i+1}-\alpha_i$. Thus, $\{\alpha_i\}_i$ is an increasing sequence, but the increase declines monotonically. We may thus take $\alpha_1-\alpha_0=1-\frac{1}{\e}$ for an upper bound on the difference between consecutive values in $\{\alpha_i\}_i$. Then we see that all our samples of $X_{W,c',c}$ take values in the interval $[-1+\frac{1}{\e},1-\frac{1}{\e}]$. Combining the above, Hoeffding's inequality gives, for any estimate $\hat{\Delta}$ of $\Delta$:
\begin{align}
\P(|\hat{\Delta}-\Delta|\ge\epsilon)\le 2\exp \left(-\frac{2\epsilon^2\ell}{(2-2/\e)^2}\right).
\end{align}
In each iteration we estimate $\Delta(W,c',c)$ based on all combinations $(c',c)$, so these are $(m-k)\cdot k$ estimations. This gives $(m-k)\cdot k\cdot\xi\alpha_k\cdot \frac{\gamma k\log k}{C_2(1-\gamma)}$ estimations in total over all iterations. The probability that for all $i$, the estimation $\hat{\Delta}_i$ of $\Delta_i$ is within $\epsilon$ of the true value, then becomes
\begin{align}
\P(\forall i: |\hat{\Delta}_i-\Delta_i|<\epsilon) \ge 1-2\exp\left(-\frac{2\epsilon^2\ell}{(2-2/\e)^2}\right)\cdot (m-k)k\cdot\xi\alpha_k\cdot\frac{\gamma k\log k}{C_2(1-\gamma)}.
\end{align}
We can now choose $\ell$ such that this it at least $1-\delta$:
\begin{align}
\begin{split}
1-2&\exp\left(-\frac{2\epsilon^2\ell}{(2-2/\e)^2}\right)\cdot (m-k)k\cdot\xi\alpha_k\cdot\frac{\gamma k\log k}{C_2(1-\gamma)} \ge 1-\delta\\
\iff &2\exp\left(-\frac{2\epsilon^2\ell}{(2-2/\e)^2}\right)\cdot (m-k)k\cdot\xi\alpha_k\cdot\frac{\gamma k\log k}{C_2(1-\gamma)}\le\delta\\
\iff &\exp\left(-\frac{2\epsilon^2\ell}{(2-2/\e)^2}\right)\le \frac{1}{2}\delta\cdot\frac{1}{(m-k)k}\cdot\frac{1}{\xi\alpha_k}\cdot\frac{C_2(1-\gamma)}{\gamma k\log k}\\
\iff & \ell\ge \frac{(2-2/\e)^2}{2\epsilon^2}\log\left(\frac{2\cdot(m-k)k\cdot\xi\alpha_k\cdot  \gamma k\log k}{C_2(1-\gamma)\delta}\right).
\end{split}
\end{align}
This means, for any $\delta>0$, any $\ell$ that is at least this size, will, with probability $1-\delta$, produce (querying $\ell$ voters) estimates that are all at most $\epsilon$ apart from the value that they estimate. We take the smallest possible value of $\ell$:
\[\ell=\left\lceil\frac{(2-2/\e)^2}{2\epsilon^2}\log\left(\frac{2\cdot(m-k)k\cdot\xi\alpha_k\cdot \gamma k\log k}{C_2(1-\gamma)\delta}\right)\right\rceil.\]

The algorithm terminates when $\hat{\Delta}(W,c',c)<\frac{C_2(1-\gamma)}{\gamma k\log k}-\epsilon$, which implies $\Delta(W,c',c)<\frac{C_2(1-\gamma)}{\gamma k\log k}$. This is the original stopping condition of the algorithm without queries, so the original performance guarantees also hold here; the approximation ratio is $(1-1/\e)\gamma$. Last we study the number of queries. In total there are at most $\xi\alpha_k\frac{\gamma k\log k}{C_2(1-\gamma)}$ iterations and per iteration we need to make $t\left\lceil\frac{m-k}{t-k}\right\rceil$ queries. Combining with the obtained value of $\ell$, we get
\begin{align}
&\xi\alpha_k\cdot\frac{\gamma k\log k}{C_2(1-\gamma)}\cdot t\left\lceil\frac{m-k}{t-k}\right\rceil\cdot\left\lceil\frac{(2-2/\e)^2}{2\epsilon^2}\log\left(\frac{2\cdot(m-k)k\cdot\xi\alpha_k\cdot \gamma k\log k}{C_2(1-\gamma)\delta}\right)\right\rceil\\
&= \xi\alpha_k\cdot\left(\frac{\gamma k\log k}{C_2(1-\gamma)}\right)^3\cdot t\left\lceil\frac{m-k}{t-k}\right\rceil\cdot\left\lceil\frac{(2-2/\e)^2}{2}\left(\frac{2\xi}{\xi-1}\right)^2\log\left(\frac{2\cdot(m-k)k\cdot\xi\alpha_k\cdot\gamma  k\log k}{C_2(1-\gamma)\delta}\right)\right\rceil \nonumber
\end{align}
queries to $A$. Asymptotically in $\,\gamma,\ \delta,\ k$ and $m$ (recall $k\ll m$), this is in $O\left(\left(\frac{\gamma k\log k}{1-\gamma}\right)^3m\log\left(\frac{m\gamma k^2\log k}{\delta(1-\gamma)}\right)\right)$.
\end{proof}

\section{Lower bound for non-adaptive algorithms in the incomplete information setting}\label{app:non-adaptive}
For the incomplete information setting, we have seen two algorithms. Both algorithms make use of information acquired during the run of the algorithm. Indeed, when selecting the query sets, we need these to contain the currently elected committee. We next present a lower bound for the query complexity of non-adaptive algorithms, which can not adapt their querying strategy to readily obtained information. The proof concerns the original setting, without matroid constraints, but the result extends to the problem over a matroid constraint, as it is a more general problem.

\begin{theorem}\label{thm:incomplete_lower_bound}For any $k\ge 2$, $\epsilon>0$, and $\gamma\in(\frac{1}{2},1)$, any non-adaptive committee selection algorithm that makes fewer than $\Omega(\frac{m^2 \epsilon}{1-\gamma})$ queries, is $(1-1/\e)\gamma$-approximate for CC with probability at most $\epsilon$.	
\end{theorem}

\begin{proof}
We denote a given non-adaptive committee selection algorithm by $\alg$. We define the class of instances that will serve as an example for the lower bound: let $k=2p+r$, where $p\in\N$ and $r\in\{0,1\}$, which means we can form any $k\in\N_{\ge 2}$ for the committee size. Partition the candidates into $C_1\cup C_2\cup\ldots\cup C_p\cup D$, with $|C_i|=\lfloor (m-r)/p\rfloor$ for all $i\in\{1,\ldots,p\}$, $D$ containing the remaining candidates. Because $\sum_{i=1}^p|C_i|=p\cdot \lfloor (m-r)/p\rfloor \le m-r$, we know $|D|\ge r$. For each $i$, we form $S_i=\{c_{i,1},\ldots,c_{i,\ell}\}\subseteq C_i$ by picking $\ell$ distinct candidates from $C_i$. The candidates are all chosen to join $S_i$ with the same probability, independently of one another, and the assignment of these candidates to the numbers $1,\ldots,\ell$ is also random. Moreover, this procedure happens independently for all $C_i$. We partition the voters into $p$ parties $P_1,\ldots,P_p$ and, if $r=1$, one party $R$ (if $r=0$ there is no set $R$). Each party $P_i$ contains $2/k$ of the voters and all voters in $P_i$ approve only of candidates within $S_i$ (so no one approves any candidate in $C_i\setminus S_i$). More specifically, for each $P_i$, we can again partition the candidates of $C_i$ into two groups with non-overlapping approval sets: candidates $c_{i,1},\ldots,c_{i,\ell-1}$, are collectively approved by $1/2$ of the voters in $P_i$ (so $1/(2k)$ voters in total) and $c_{i,\ell}$, that we call the \emph{designated candidate}, is approved by the other half of $P_i$. The party $R$, if it exists, contains $1/k$ of the voters, all of whom approve precisely one candidate $d\in D$.

We say that $\alg$ \emph{covers} a set of candidates $S$ if $\alg$ ever submits a query set $Q$ for which $|Q\cap S|\ge 2$. This means that, if for any of the parties $P_i$, $\alg$ fails to cover $S_i$, all $\ell$ of these candidates are indistinguishable to $\alg$, given the definition of our query model. Indeed, all candidates within $S_i$ are approved by the same number of voters, and we can only start to distinguish the structure of $S_i$ if we combine $c_{i,\ell}$ and $c_{i,j}$ for some $j\in\{1,\ldots,\ell-1\}$ in one query. Then, if a set $S_i$ is not covered, the best $\alg$ can do is pick candidates from $S_i$ randomly, since  without covering $S_i$ it is still possible to distinguish which elements form the set $S_i$. 

For each $i$, writing $k_i$ for the number of candidates $\alg$ selects from $C_i$, $\alg$ picks the designated candidate with probability at most $k_i/\ell$. Suppose that for all $i$, $\alg$ picks some candidate $c_{i,j}$ for $j\in\{1,\ldots,\ell-1\}$ and also picks the single approved candidate in $R$. We are allowed to assume this since this assumption facilitates $\alg$. This already gives a CC-score of $\frac{p+1}{2p+1}$. Then, in order to attain a CC-score of $\gamma(1-1/\e)$, $\alg$ must still pick $\gamma(1-1/\e)(2p+1)-(p+1)$ designated candidates. Then if $t=1$, it is impossible for $\alg$ to cover any of the $S_i$ and we attain a ratio of $\gamma(1-1/\e)$ only if $\alg$ is lucky enough to select at least $\left\lceil \left(1-\frac{1}{\e}\right)\gamma  + \left(\left(1-\frac{1}{\e}\right)2\gamma - 1\right)p - 1\right\rceil=: q$ of the designated candidates without having any information about which ones to pick. Clearly, for each $C_i$, the probability of picking the correct candidate is at most $p/\ell$, as $p$ is the total number of candidates we can still pick, and $\ell$ is the size of each set. The probability to pick at least $q\geq 1$ designated candidates is at most the probability to pick a single one: $p/\ell$.



Now suppose that $t>1$ and so we might be able to cover some $S_i$. Moreover, suppose that $\alg$ knows the partition of candidates into $C_1\cup\ldots\cup C_p\cup D$ (so it will query within these sets) and knows the distribution of approvals within the sets $C_i$ (but it doesn't know how this corresponds to the particular candidates). Moreover, we assume we can make at most $U\cdot m^2$ queries within each set $C_i$, for some constant $U$.

For any party $P_i$ consider any set $S\subset C_i$ of size $\ell$ that is covered by a query set $Q$ of size $t$. As soon as $\ell>t$ (so when $m$ is large enough), we can split $S$ into two parts: a part of size $2\le j\le t$ that is \emph{actually} covered by $Q$, and a part of size $\ell-j$ that is outside of $Q$. The number of sets $S\subset C_i$ of size $\ell$ that a single query can cover, then becomes
\begin{align}
\sum_{j=2}^t\binom{t}{j}\binom{\lfloor(m-r)/p\rfloor -t}{\ell-j}\le 2^t \binom{m/p}{\ell-2}\le 2^t \left(\frac{m}{p}\right)^{\ell-2}
\end{align}
if $\ell-2\le \frac{m}{2p}$, because $\binom{n}{k}$ grows in $n$, and grows in $k$ until $k=n/2$. Note that we require $m$ to be large enough for this. With our $Um^2$ queries, at most
\begin{align}
Um^2 \cdot 2^t \left(\frac{m}{p}\right)^{\ell-2} = Um^\ell 2^tp^{2-\ell}
\end{align}
sets of size $\ell$ within $C_i$ can be covered in total. For $m$ large enough, there are
\begin{align}
\binom{\lfloor(m-r)/p\rfloor}{\ell}\ge \frac{(\lfloor(m-r)/p\rfloor-\ell)^\ell}{\ell!} \ge\frac{1}{\ell!}\left(\frac{m}{2p}\right)^\ell
\end{align}
sets of size $\ell$ in $C_i$ in total, and since all sets were chosen to join $S_i$ with the same probability, $\alg$ covered $S_i$ with probability at most
\begin{align}
\frac{\ell!(2p)^\ell Um^\ell 2^tp^{2-\ell}}{m^\ell}= \ell!2^{\ell+t} U p^{2}.
\end{align}
Suppose that $\alg$ picked at least one of the candidates $c_{i,j}$ for $j\in\{1,\ldots,\ell-1\}$ for all $i$, then, like before, we need to pick at least $q$ designated candidates in order to attain the desired ratio. $\alg$ can do this by covering no $S_i$ and picking all $q$ designated candidates by chance, which happens, like proven above, with probability at most $p/l$. $\alg$ can also cover one of the sets $S_i$ and pick $q-1$ designated candidates by chance, or any other combination is possible up to covering all sets $S_i$. Then, for $t>1$, 
\begin{align}
 &\P(\alg \text{ selects at least } q\text{ designated candidates}) = \sum_{j=q}^p  \P(\alg \text{ selects } j\text{ designated candidates})\\
 &\le \sum_{j=q}^p \sum_{h=0}^j \P\left(\{\alg \text{ selects } j-h\text{ designated candidates by chance}\}\bigcap\, \{\alg\text{ covers }h\text{ of the sets } S_i\}\right),\nonumber
\end{align}

where the inequality is due to the fact that covering a set $S_i$ is necessary to distinguish its designated candidate, but might not be sufficient. For $h=0$, we have bounded the probability above as $p/\ell$. Each of the other at most $(p-q)p$ probabilities is bounded from above by the probability of covering one set $S_i$, which is at most $\ell!2^{\ell+t} U p^{2}$. If we fix $U=\frac{\delta}{\ell!2^{t+\ell}p^3(p-q)}$ this becomes at most $\delta$. The total probability is then at most $\frac{p}{\ell} + \delta$. Recall that we want this to be at most $\epsilon$, which we can achieve, for instance, by fixing both terms to $\epsilon/2$. Since $\delta$ determines the query complexity, we can choose $\delta$, so it suffices to show that $\frac p\ell \leq \frac\epsilon2$. Since $p=2k-1$, this condition is met when $l\geq \frac{k-1}{\epsilon}$.

As mentioned above, we have at most $Um^2$ queries and $U=\frac{\epsilon/2}{\ell!2^{t+\ell}p^3(p-q)}$, where $p=(k-1)/2$. We want to denote the scaling of the complexity in terms of $\gamma$ and $\epsilon$. Since $p=(k-1)/2$ and $q = \left\lceil \left(1-\frac{1}{\e}\right)\gamma  + \left(\left(1-\frac{1}{\e}\right)2\gamma - 1\right)p - 1\right\rceil$, we obtain 
\begin{align}
        \frac{1}{p-q} 
        \leq \frac{2}{1+k-2(1 + ((1-\frac{1}{\e})\gamma - \frac{1}{2})k )} 
        = \frac{2}{2k(1-(1-\frac{1}{\e})\gamma) - 1} 
        \leq \frac{1}{k(1-\gamma)} ,
\end{align}
where the last step required $2k\gamma/2 \geq 1/2$, which is true in case $k\geq 2, \gamma>1/2$. It is not hard to see that $\frac{1}{\ell!2^{t+\ell}p^3} \leq 1$, so that the final query complexity becomes $\frac{m^2 \epsilon}{1-\gamma}$.
\end{proof}

Theorem~\ref{thm:incomplete_lower_bound} motivates the study of adaptive algorithms. It might be possible to extend this result to a larger exponent of $m$, as done in \cite{halpern_representation_2023}), by adapting the instance. We leave this for future work.

\section{Inaccurate information: proofs of Theorem~\ref{thm:inaccurate_upper_bound} and Theorem~\ref{thm:inaccurate_lower_bound}}\label{app:inaccurate}
We prove the upper and lower bound on achieving the optimal approximation ratio for CC in the $p$-inaccurate query model. We stated these results as Theorem~\ref{thm:inaccurate_upper_bound} and Theorem~\ref{thm:inaccurate_lower_bound}, respectively.

Recall that for accurate information, we write $A(i,j)$ to query the opinion of voter $v_i\in V$ about candidate $c_j\in C$, where the outcome is 1 in case of approval and 0 in case of disapproval. In the $p$-inaccurate query model, we instead query $A_p(i,j)=A(i,j)\oplus X$ where $X\sim \text{Bernoulli}(p)$, so that the outcome is flipped with probability $p$. This error is added independently for each pair $(i,j)$ and each repeated sampling. For the upper bound, we give an algorithm which recovers the optimal approximation ratio with high probability. Our goal is to minimize the number of queries to $A_p(i,j)$. 

\begin{algorithm}[h]
	\SetAlgoNoLine
	\KwIn{Numbers $n,m,k\in\N$ with $k\le m$, set $V$ of $n$ voters, set $C$ of $m$ candidates, parameters $p\in(0,\frac{1}{2}),\,\delta>0$.}
	\KwOut{Committee $W\subseteq C$ of size $k$.}
    For all $v_i\in V,\, c_j\in C$, perform query $A_p(i,j)$ $U=\left\lceil 2\frac{\log((nm)/\delta)}{1/\log(4p(1-p))}\right\rceil$ times and take the majority winner as the `true' query response\;
    Let $W=\{ \}$ be an empty set\;
    \For{$i=1,\ldots,k$,}{
		\begin{enumerate}
			\item For all $c\notin W$, determine $\Delta(W,c)$.
			\item Add to $W$ any  $c'\in\argmax_{c\notin W}\,\Delta(W,c)$.
		\end{enumerate}
	}
    \caption{\textsc{greedy-inaccurate}}
    \label{alg:greedy-inaccurate}
\end{algorithm}

\begin{reptheorem}{thm:inaccurate_upper_bound}
Let $p\in(0,\frac12)$, $\delta>0$, $n,m\in\mathbb{N}$. Then there exists an algorithm that is $(1-1/e)$-approximate for CC in the $p$-inaccurate query model w.p. at least $1-\delta$ and with query complexity $O(nm\log(nm/\delta))$.
\end{reptheorem}
\begin{proof}
We adapt Algorithm~\ref{alg:greedy}, obtaining Algorithm~\ref{alg:greedy-inaccurate}. While Algorithm~\ref{alg:greedy} simply queries $A(i,j)$ once for each of the $nm$ voter-candidate pairs, Algorithm~\ref{alg:greedy-inaccurate} instead performs each query $A_p(i,j)$ some $U>1$ number of times and takes the majority winner to be the `true' value, breaking ties arbitrarily. We now show that, taking $U$ as in Algorithm~\ref{alg:greedy-inaccurate}, each of the $nm$ majority votes are correct with probability at least $1-\frac{\delta}{nm}$. Then, by independence, all together are correct with probability $1-\delta$, in which case we obtain a $(1-1/\e)$-approximate solution. Each majority vote fails if a strict majority of the $U$ samples is incorrect, and we can bound this failure probability using a Chernoff bound: Take $\Omega = \{0,1\}^U$, $\mathcal{F}=\mathcal{P}(\Omega)$ and $\P:\mathcal{F}\to[0,1]$ that satisfies $\P(\{\mathbf{v}\})=(1-p)^y p^{U-y}$ when $\mathbf{v}$ contains $y$ 1's and $U-y$ 0's. Accordingly, write $\mathbf{v}_{i,j}$ for a vector that is the outcome of $U$ throws of $(i,j)$, where entry 0 means the outcome is corrupted and 1 means the outcome is correct. Then take $X:\Omega\to\N: \mathbf{v}_{i,j} \mapsto \sum_{l=1}^U v_{i,j,l}$ for the random variable that maps a vector to the number of correct outcomes and write $\bar{X}:= U-X$ for the number of failures. A Chernoff bound then tells us that
$\P(\bar{X}\ge \frac{1}{2}U)\le \exp(-UD(\frac{1}{2}||p))$ \cite{gordon_tutorial_1989} and we can reformulate to
\begin{align}
\begin{split}
\exp\left\{-UD\left(\frac{1}{2}||p\right)\right\} &:= \exp\left\{-U\left(\frac{1}{2}\log\left(\frac{1/2}{1-p}\right) + \frac{1}{2}\log\left(\frac{1/2}{p}\right)\right)\right\}\\
&= \exp\left\{-U\frac{1}{2}\log\left(\frac{1/2}{1-p}\right)\right\}\exp\left\{-U\frac{1}{2}\log\left(\frac{1/2}{p}\right)\right\}\\
&=\left(\frac{1}{4(1-p)p}\right)^{-\frac{1}{2}U}.
\end{split}
\end{align}
We determine the value of $U$ so that the above is at most $\frac{\delta}{nm}$: 
\begin{align*}
&\left(\frac{1}{4(1-p)p}\right)^{-\frac{1}{2}U} \le \frac{\delta}{nm} \iff -\frac{1}{2}U \le \log_{\frac{1}{4p(1-p)}}\left(\frac{\delta}{nm}\right) \\
&\iff U \ge -2 \log_{\frac{1}{4p(1-p)}}\left(\frac{\delta}{nm}\right) = -2\frac{\log \frac{\delta}{nm}}{\log\frac{1}{4p(1-p)}}  = 2\frac{\log \frac{nm}{\delta}}{\log\frac{1}{4p(1-p)}}.
\end{align*}
Fixing $\left\lceil 2\frac{\log((nm)/\delta)}{1/\log(4p(1-p))}\right\rceil$ thus guarantees that all estimates are correct, and thus the approximation ratio is achieved with probability at least $1-\delta$.
\end{proof}

Although the above proof involved an adaptation of the Greedy algorithm, we note that we could have just as easily adapted the local search algorithm. Both algorithms will query all $n$ voters about all $m$ candidates. To deal with the $p$-inaccurate information setting, the above proof simply shows that repeating each of the $nm$ queries $U=\log(nm/\delta)$ times guarantees that the $nm$ queries are all correct with probability $1-\delta$. This guarantee thus holds independently of what the algorithm does with the query outcomes, and also holds for the local search algorithm. This also implies that the above upper bound holds for the generalized problem of optimizing over a matroid contraint.

\begin{reptheorem}{thm:inaccurate_lower_bound}
Let $p\in(0,\frac12)$, $\delta>0$, $n,m\in\mathbb{N}$. Then any algorithm that is $(1-1/e)$-approximate for CC in the $p$-inaccurate query model w.p. at least $1-\delta$ has expected query complexity $\Omega(nm\log(1/\delta))$.
\end{reptheorem}
To establish the theorem, we adapt the proof of Theorem 9 in \cite{schafer_complexity_2024}, using a result from multi-armed bandit theory \cite{lattimore_bandit_2020} to lower bound the query complexity. The Kullback-Leibler divergence between two probability distributions is defined as $D(P||Q)=\sum_{x} P(x)\log\frac{P(x)}{Q(x)}$ for discrete probability functions $P$ and $Q$, as long as, for all $x,\, Q(x)=0\implies P(x)=0$ (otherwise we define it as $+\infty$). This reduces to $d(a,b)=a\log\frac{a}{b}+(1-a)\log\frac{1-a}{1-b}$ for two Bernoulli distributions with parameters $a,b\in(0,1)$.

\begin{proof}
Recall that we query voter $v_i\in V$ about candidate $c_j\in C$ using $A_p(i,j)=A(i,j)\oplus X$ where $X\sim Bernoulli(p)$. We will view each query $A_p(i,j)$ as an arm $a\in\mathcal{A}$ in a multi-armed bandit problem, and write $y_t:=A_p(i,j)_t$ for the outcome of the $t$-th query. When $A(i,j)=0$. Then $A_p(i,j)\sim \text{Bernoulli}(p)$. Conversely, when $A(i,j)=1$ then $A_p(i,j)=\text{Bernoulli}(1-p)$. We let $P_a$ denote the Bernoulli distribution associated with action $a$: for any query index $t$ and any $y\in\{0,1\}$ we have $P_a(y)=\P(y_t=y|a_t=a)$. Let us refer to an arbitrary (possibly randomized) algorithm for the problem by $\alg$. We write $\E^{\nu}$ and $\P^{\nu}$ for the expectation and probability when the underlying instance is $\nu$. Write $L_\nu$ for the set of optimal committees for instance $\nu$. By assumption we have $\P^\nu(\alg(\nu)\in L_\nu)\ge 1-\delta$ for any instance $\nu$. We denote the random variable that portrays the number of queries taken when $\alg$ terminates, by $\tau$. We use the following lemma from multi-armed bandit theory, that is mentioned e.g. in \cite{lattimore_bandit_2020}, Chapter 15 (Lemma 15.1 and more specifically Exercise 15.7).
\begin{lemma}\label{lemma_KL}
Let $\nu$ and $\nu'$ be any two bandit instances defined on the same set of arms $\mathcal{A}$, with corresponding observation distributions $\{P_a\}_{a\in\mathcal{A}}$ and $\{P'_a\}_{a\in\mathcal{A}}$. Let $\tau$ be the total number of queries made when the algorithm terminates, and let $\mathcal{E}$ be any probabilistic event that can be deduced from the resulting history $(a_1,y_1,\ldots,a_\tau,y_\tau)$, possibly with additional randomness independent of that history. Then, for $T_a$ the number of times action $a$ is queried up to termination index $\tau$, we have
\begin{align}
\sum_{a\in\mathcal{A}}\E^\nu[T_a]D(P_a||P'_a)\ge d(\P^\nu(\mathcal{E}),\P^{\nu'}(\mathcal{E})).
\end{align}
\end{lemma}
 
For any $k\ge 2$, consider an instance $\nu$ with $m$ candidates and $n$ voters. We partition the voters into $m$ disjoint ``parties" $P_1,\ldots,P_{m}$. Voters in party $P_i$ approve precisely $c_i$. Candidates $c_1,\ldots,c_{m-k}$ (together $C_1$) all have one approving voter. These $m-k$ voters (together $V_1$) are alone in their party. Candidates $c_{m-k+1},\ldots,c_{m}$ (together $C_2$) all have $u>\frac{1}{1-\frac{1}{\e}\cdot k}$ approving voters (together $V_2$), so that the only $(1-\frac{1}{\e})$-approximate solution consists of the $k$ candidates in $C_2$:
\begin{align}
(k-1)u+1 < \left( 1-\frac{1}{\e} \right)ku \iff \left(\frac{1}{\e} k-1 \right) u < -1\iff u>\frac{1}{1-\frac{k}{\e}}.
\end{align}
This shows even holding in the committee $k-1$ candidates from $C_2$ and one candidate from $C_1$, instead of $k$ candidates from $C_2$, yields an unforgivable loss in score. Finally we have $V_3$, a set of voters that approve of no candidate (homeless voters). For any $m,n$, we can find constants $D,E\in (0,1)$ such that $|C_1|\ge D\cdot m$ (which means we need $0\le m-Dm-k$) and $(1-E)n -ku \ge |C_1|$ so that $|V_3|\ge E\cdot n$, so let's pick any such constants $D$ and $E$. Now consider the instance $\nu '$ that is obtained by taking, $x:=\lceil k(\frac{u}{1-1/\e}-1)\rceil$ homeless voters and distributing them (as) evenly (as possible) over $k$ candidates $c_{1},\ldots,c_k\in C_1$. Other than that, all preferences remain unchanged. With this value of $x$, we have $x+k>\frac{ku}{1-1/\e}$, which means that $\nu '$ has a unique $(1-1/\e)$-approximate solution in $c_1,\ldots,c_k\subseteq C_1$, whereas the unique $(1-1/\e)$-approximate solution of $\nu$ was the set $C_2$. Since both instances have only one optimal solution, and these differ, we have $L_{\nu}\cap L_{\nu'}=\emptyset$.

\begin{figure}
    \begin{subfigure}[h]{0.44\linewidth}
        \includegraphics[width=\linewidth]{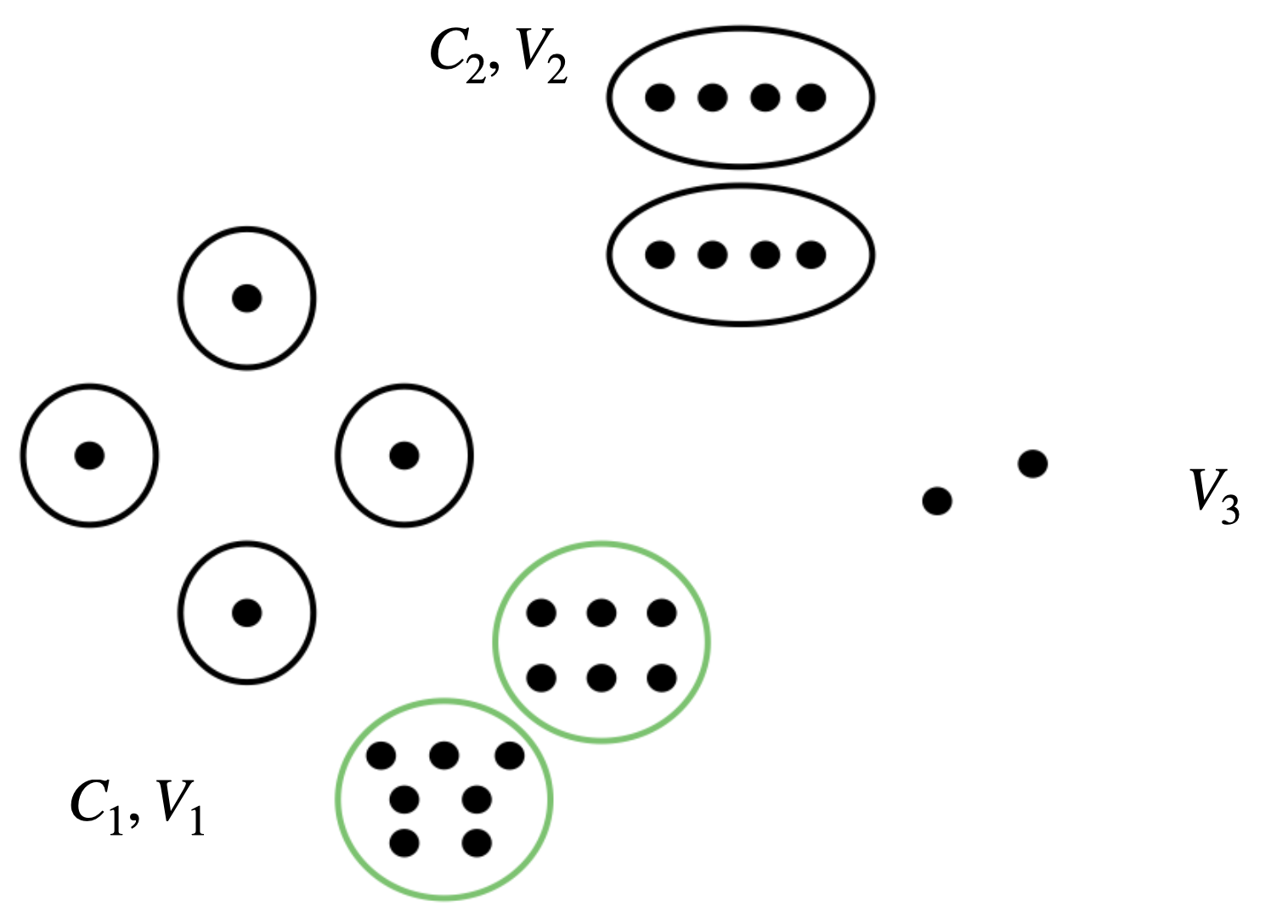}
    \caption{Instance $\nu$}
    \end{subfigure}
\hfill
    \begin{subfigure}[h]{0.44\linewidth}
        \includegraphics[width=\linewidth]{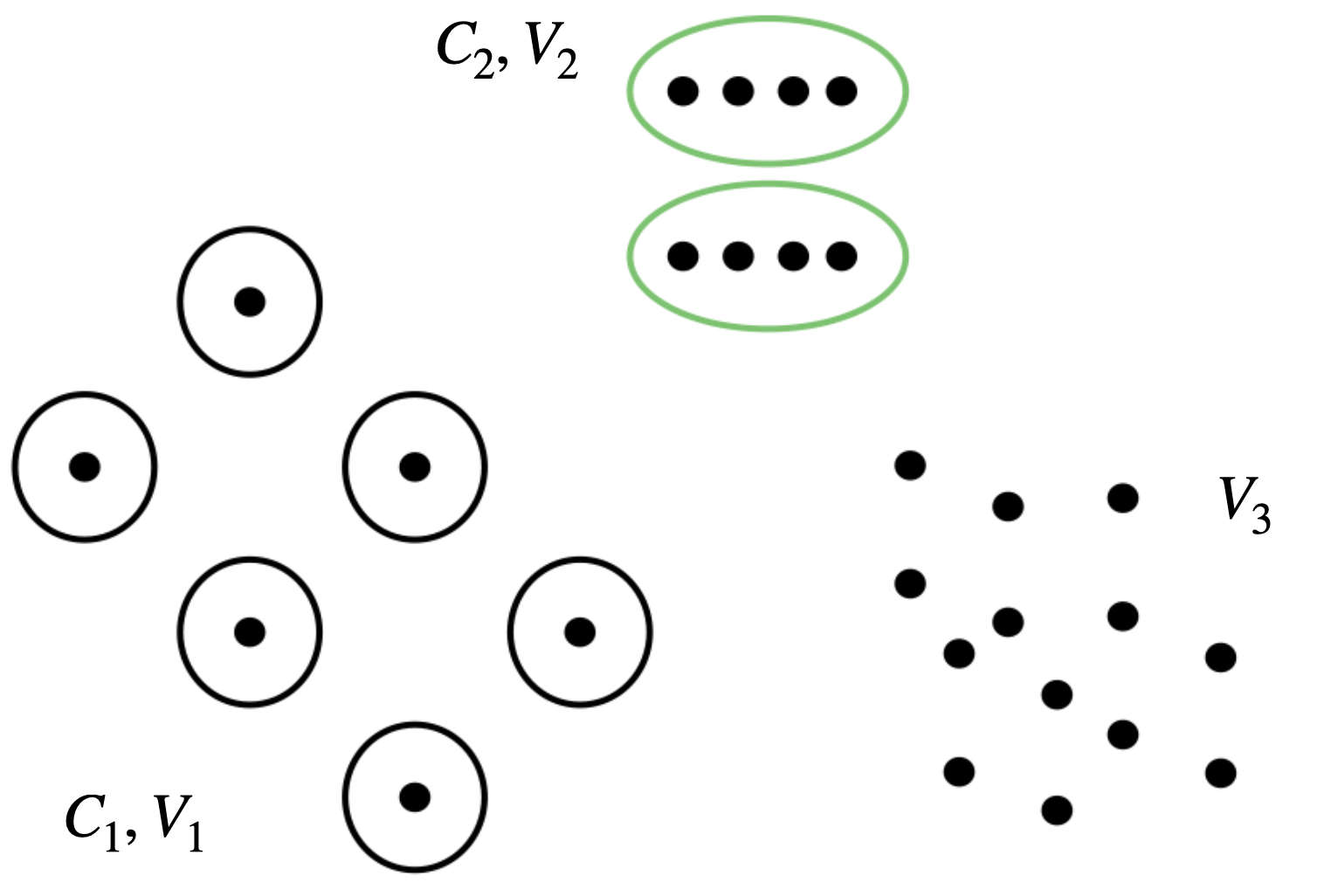}
    \caption{Instance $\nu'$}
    \end{subfigure}
\caption{When $k=2$, both instances have only one $(1-1/\e)$-approximate solution and these differ.}
\label{fig:lower_bound}
\end{figure}

Now, regarding the inequality in Lemma \ref{lemma_KL}, note the following. Firstly, $D(P_a||P'_a)=0$ for all $a\notin \{(i,j):i\in\{1,\ldots,x\}, j\in\{1,\ldots,k\}\}$. When $D(P_a,P'_a)\neq 0$, it equals $d(1-p,p)$. Let $\mathcal{E}$ be the event that $\alg$ outputs (the only element in) $L_{\nu'}$. Then, by assumption, $\P^{\nu}(\mathcal{E})\le\delta$ and $\P^{\nu'}(\mathcal{E})\ge 1-\delta$. Monotonicity of the discrete KL function $d$ implies that $d(\P^\nu(\mathcal{E}),\P^{\nu'}(\mathcal{E}))\ge d(\delta,1-\delta)$.

So now we can rewrite the inequality in Lemma \ref{lemma_KL} as 
\begin{align}
\left(\sum_{i=1}^{x}\sum_{j=1}^k\E^\nu[T_{(i,j)}]\right) \cdot d(1-p,p)\ge d(\delta,1-\delta),
\end{align}
or equivalently, $\left(\sum_{i=1}^{x}\sum_{j=1}^k\E^\nu[T_{(i,j)}]\right)\ge \frac{d(\delta,1-\delta)}{d(1-p,p)}$. We can make such an adapted instance $\nu'$ for at least $\lfloor\frac{En}{x}\rfloor$ disjoint sets of $x$ voters in $V_4$ (we gather these sets in a set we call $V_{4,x}$) and at least $\lfloor\frac{Dm}{k}\rfloor$ disjoint sets of $k$ candidates in $C_1$ (we gather these sets in a set we call $C_{1,k}$):

\begin{align}
\sum_{C_{1,k}}\sum_{V_{3,x}}\left(\sum_{i=1}^{x}\sum_{j=1}^k\E^\nu[T_{(i,j)}]\right)\ge \left\lfloor\frac{En}{x}\right\rfloor\left\lfloor\frac{Dm}{k}\right\rfloor \frac{d(\delta,1-\delta)}{d(1-p,p)}\ge \left\lfloor\frac{E}{x}\right\rfloor\left\lfloor\frac{D}{k}\right\rfloor\cdot mn\cdot \frac{d(\delta,1-\delta)}{d(1-p,p)}.
\end{align}
We rewrite the left hand side as 
\begin{align}
\sum_{C_{1,k}}\sum_{V_{3,x}}\left(\sum_{i=1}^{x}\sum_{j=1}^k\E^\nu[T_{(i,j)}]\right)
\le \sum_{v_i\in V}\sum_{c_j\in C}\E^\nu[T_{(i,j)}] = \E^\nu\left[\sum_{v_i\in V}\sum_{c_j\in C}T_{(i,j)}\right] = \E^{\nu}[\tau],
\end{align}

so we rephrase to $\E^{\nu}[\tau]\ge \left\lfloor\frac{E}{x}\right\rfloor\left\lfloor\frac{D}{k}\right\rfloor\cdot mn\cdot \frac{d(\delta,1-\delta)}{d(1-p,p)}$. For the right hand side, note that $p$ is a fixed constant in $(0,\frac{1}{2})$ and $d(\delta,1-\delta)\in\Omega(\log(1/\delta))$ since $\delta\le \frac{1}{2}-c$, hence, this instance $\nu$ serves to illustrate the lower bound.
\end{proof}

\section{Discussion on the experiments}\label{app:experiments}
We discuss two aspects of the experiments. First, the pre-processing applied to the approval elections obtained from Polis, and second, how we determined the parameter $\phi$ that we used to instantiate the $(q,\phi)$-resampling model of \cite{szufa_how_2022} to generate artificial approval elections.

\paragraph{Pre-processing of the data.}
We did four pre-processing steps. First, of all 20 accessible Polis datasets, we remove the \texttt{london.youth.policing} dataset due to a low number of participants, and the \texttt{bg2025-volunteers} dataset because it contains far more statements than voters (all other sets have far fewer statements than voters, on average $m/n=0.4$). Second, we removed statements with more than half of voters approving, as they made attaining a high CC-score trivial in some datasets (and are generally relatively weak statements). Third, we removed voters that did not vote on any statement as well as those who did not approve of any statements (as they are unsatisfiable). 

Fourth and last, as mentioned in the main text already, we need to complete the data. On average across the 18 datasets, voters were presented with 18.2\% of the statements. Their responses were 56.3\% approvals, 27.3\% disapprovals and 16.4\% neutral votes. Since we require every query to be answered either as approval or disapproval we need to fill in unanswered queries and queries answered neutrally. We set all such votes to disapprove, which results in a worst-case for the CC-score of the algorithms. 

Our motivation for this is twofold. First, we are not convinced the partial ballot is representative of the full ballot, making extrapolation a doubtful strategy. Indeed, in Polis, the query selection is random but weighed, taking into account, among other things, how divisive and recent a statement is, and how often it has answered neutrally. Moreover, often up to a third of voters only respond to a single query (often their own statement), while a small portion of voters vote on hundreds of statements. Second, since the majority of cast votes are approvals, extrapolation will make attaining a high CC-score trivial. Taking CC as diversity measurement, as we do, this observation shows a diverse  committee is practically achievable in the context of Polis (assuming the partial ballots are representative). But this also means that this setting is not of interest to study empirically. By setting all unknowns to disapprove, we obtain the worst-case CC-score across all extensions of the partial ballot (therefore, including the real extension), recovering an interesting empirical case-study: if we achieve good CC-score here, we also can in practice. However, note that although this is the worst-case for the CC-score of each individual algorithm, it is not necessarily worst-case with respect to the ratio between complete/accurate and incomplete/inaccurate algorithms. This would only be true in case the relative decrease in the incomplete/inaccurate CC-score is smaller than that of the complete/accurate CC-score.

\paragraph{Overview of Polis datasets}
Figure~\ref{fig:plot2} shows the performance of four algorithms on the 18 Polis datasets. The enumeration below indicates which datasets is mapped to each number.

\begin{enumerate}
    \item ``austria-climate.9xnndurbfm.2022-07-07"
    \item ``austria-climate.7z7ejpbmv5.2022-08-08"
    \item ``american-assembly.bowling-green"
    \item ``austria-climate.5twd2jsnkf.2022-08-08"
    \item ``austria-climate.5tzfrp5eaa.2022-07-07"
    \item ``march-on.operation-marchin-orders"
    \item ``ssis.land-bank-farmland.2rumnecbeh.2021-08-01"
    \item ``scoop-hivemind.affordable-housing"
    \item ``scoop-hivemind.biodiversity"
    \item ``canadian-electoral-reform"
    \item ``scoop-hivemind.ubi"
    \item ``austria-climate.2vkxcncppn.2022-07-07"
    \item ``football-concussions"
    \item ``scoop-hivemind.taxes"
    \item ``vtaiwan.uberx"
    \item ``brexit-consensus"
    \item ``15-per-hour-seattle"
    \item ``scoop-hivemind.freshwater"
\end{enumerate}

\paragraph{Determining $\phi$ of the Polis data.}
To determine $\phi$, we compute the approvalwise vector \cite[Definition 3]{szufa_how_2022} for each election, and then compute the approvalwise distance \cite[Definition 4]{szufa_how_2022} between each election and the limit of the approvalwise vector \cite[Section 5]{szufa_how_2022} of $(q,\phi)$-resampling elections for $q=0.0891$ and $\phi=0,1/100,2/100, \dots,1$. For each election, we determine the value of $\phi$ that minimizes this distance and then average over the 18 Polis datasets. This yields $\phi=0.693$.
\end{document}